\documentclass[11pt,a4paper]{article}

\usepackage[tbtags]{amsmath}

\usepackage[DIV=12]{typearea}

\usepackage{lineno}  
\usepackage[
oldstylenums,
]{kpfonts}  

\usepackage[scr=boondoxupr]{mathalfa} 
\makeatletter 
\newcommand\amp{{\if\f@series b\usefont{T1}{cmr}{bx}{it}\else\usefont{T1}{cmr}{m}{it}\fi\& }}
\makeatother

\usepackage[english]{babel}

\usepackage{scrextend} 

\usepackage{setspace}

\usepackage{physics}
\usepackage{mathdots}

\usepackage{url} 
\usepackage[pagebackref, bookmarks=false,]{hyperref} 
\hypersetup{
	pdftitle={Superintegrability of  discrete-time rational Ruijsenaars-Schneider model and deformed polynomial symmetry algebras},
    pdfauthor={Pavel Drozdov},
    colorlinks,
    citecolor=brightviolet,
    filecolor=black,
    linkcolor=brightviolet,
    urlcolor=brightviolet,
    final=true
} 

\usepackage[dvipsnames]{xcolor} 
\definecolor{my-background}{RGB}{255, 242, 204}
\definecolor{brightviolet}{RGB}{148,0,211}  

\usepackage{tikz}
\usetikzlibrary{arrows,matrix,positioning,arrows.meta}
\usetikzlibrary{decorations.pathreplacing, calc,fit}
\usepackage{tikzsymbols} 
\usepackage{blkarray}
\usepackage{tikz-cd} 
\usepackage{microtype}

\usepackage{anyfontsize}

\allowdisplaybreaks  

\usepackage{mathbbol}
\usepackage{braket}
\usepackage{tensor} 

\usepackage{stmaryrd} 
\SetSymbolFont{stmry}{bold}{U}{stmry}{m}{n}

\DeclareFontFamily{U}{matha}{\hyphenchar\font45}
\DeclareFontShape{U}{matha}{m}{n}{
      <5> <6> <7> <8> <9> <10> gen * matha
      <10.95> matha10 <12> <14.4> <17.28> <20.74> <24.88> matha12
}{}
\DeclareSymbolFont{matha}{U}{matha}{m}{n}

\DeclareMathSymbol{\Asterisk}{\mathbin}{matha}{"06}
\renewcommand{\ast}{\Asterisk}

\usepackage[tiny]{titlesec}
\usepackage{titletoc}  
\titleformat{\section}
  {\normalfont \bfseries \scshape }{\thesection. }{0.1em}{}

\titlecontents{section}[0pt]{\addvspace{1ex}}{\thecontentslabel. \normalfont \bfseries \scshape}
{\Large\bfseries}{\hspace{1em plus 1fill}\large\bfseries \contentspage}

\titleformat{\subsection}
  {\normalfont \bfseries}{\thesubsection. }{0.1em}{}
  \titleformat{\subsubsection}
  {\normalfont \small \bfseries}{\thesubsubsection. }{0.1em}{}
  
  \renewcommand{\thesection}{\Roman{section}}

\newcommand{\gl}{\mathfrak{gl}}

\DeclareMathOperator{\diag}{diag}

\renewcommand{\imath}{\mathrm{i}}

\renewcommand{\vec}[1]{\boldsymbol{#1}} 

\renewcommand{\bar}{\overline}
\newcommand{\vecbar}[1]{\vec{\bar{#1}}}

\newcommand{\ubar}[1]{\mkern 1.5mu\underline{\mkern-1.5mu#1\mkern-1.5mu}\mkern 1.5mu}

\newcommand{\G}{\widetilde{G}}

\usepackage{bbold} 

\newcommand{\R}{\mathbb{R}}
\newcommand{\Z}{\mathbb{Z}}

\newcommand{\Alg}{\mathcal{A}}

\newcommand{\1}{\mathbb{1}}

\newcommand{\M}{\mathcal{M}}

\newcommand{\genset}{\mathcal{S}}

\newcommand{\calH}{\mathcal{H}}
\newcommand{\calP}{\mathcal{P}}
\newcommand{\calB}{\mathcal{B}}

\newcommand{\Bell}{B}

\newcommand{\K}{\widetilde{K}}

\newcommand{\CM}{\text{CM}}

\newcommand{\hyphen}{\hbox{--}}
\newcommand{\ie}{\textit{i.e.}}
\newcommand{\mm}{\textit{mutatis mutandis}} 
\newcommand{\eg}{\textit{e.g.}}
  
\newcommand{\cf}{\textit{cf.}}

\usepackage[shortlabels]{enumitem} 

\setlist[enumerate,1]{label={(\roman*)}}
 
\setlist[itemize,2]{label={$\circ$}}

\setlist[itemize]{itemsep=2pt, parsep=0pt, topsep=2pt}

\usepackage[normalem]{ulem}  

\usepackage{amsthm}
\usepackage{thmtools}
\usepackage[]{cleveref} 

\crefname{prop}{proposition}{propositions} 
\Crefname{equation}{Eq.}{Eqs.}

\newtheorem{theorem}{Theorem}[section]
\newtheorem*{theorem*}{Theorem}

\newtheorem{corollary}{Corollary}[theorem]

\newtheorem{prop}[theorem]{Proposition}
\theoremstyle{definition}
\newtheorem{definition}[theorem]{Definition}

\theoremstyle{remark}
\newtheorem{remark}[theorem]{Remark}
\newtheorem{example}[theorem]{Example}

\makeatletter
\renewcommand{\fnum@figure}{Fig.\,\thefigure}

\makeatletter
\renewenvironment{proof}[1][\relax]{\par
  \pushQED{\qed}%
  \normalfont \topsep6\p@\@plus6\p@\relax
  \trivlist
  \item[\hskip\labelsep\itshape
    \ifx#1\relax \proofname\else\proofname{} of 
    #1\fi\@addpunct{.}]\ignorespaces
}{%
  \popQED\endtrivlist\@endpefalse
}
\makeatother

\usepackage[auth-sc,affil-it]{authblk}

\title{Superintegrability of  discrete-time rational Ruijsenaars--Schneider model and deformed \\ polynomial symmetry algebras}

\author[1,2,$\ast$]{Pavel Drozdov}

\makeatletter
\renewcommand\AB@affilsepx{\\\vspace{0.5em}}
\makeatother

\affil[1]{
  Universit\`a degli Studi di Udine, 33100 Udine, Italy}    
\affil[2]{INFN Sezione di Trieste, 34127 Trieste, Italy}

\bgroup

\affil[ ]{\footnotesize \ttfamily \emph{e-mail:}
     \emph{\href{mailto:drozdov.pavel@spes.uniud.it}{drozdov.pavel@spes.uniud.it} 
      } }
\egroup 

\date{\vspace{-7ex}} 

\numberwithin{equation}{section}

\usepackage{lipsum}

\begin{document}
 
\maketitle 
\begin{abstract}

\noindent We explicitly construct the additional integrals of motion, ensuring maximal superintegrability of the discrete-time rational Ruijsenaars--Schneider model. Using them, we investigate the algebraic aspects of superintegrability in both continuous- and discrete-time settings. In particular, we determine the complete structures of the polynomial symmetry algebras associated with both the rational Ruijsenaars--Schneider model and its discretization. We demonstrate that discretization leads to a nontrivial deformation of the continuous symmetry algebra with respect to the discretization parameter, thereby extending recent analogous results from the rational Calogero--Moser system to its relativistic generalization.  
  \end{abstract}

\tableofcontents

{\renewcommand{\thefootnote}{$\ast $}\footnotetext{Visiting Osaka Metropolitan University, 558-8585 Osaka, Japan}} 

\newpage

\section{Introduction}
\label{sec:intro}

Among all generic dynamical systems, certain models are notable for their regular and predictable behavior --- a property formalized as integrability. This regularity is deeply connected with the existence of underlying symmetries (integrals of motion) of a given model. In this context, \emph{(maximally) superintegrable systems} admit the maximum possible number of symmetries, providing analytic and algebraic solvability~\cite{MPW}. Famous examples of superintegrable models include the isotropic harmonic oscillator, the Kepler--Coulomb model, and the Calogero--Moser system. Although superintegrable models form a narrow subclass of all dynamical systems, they frequently appear across applications: from theoretical physics to the theory of special functions; see, \eg{,}~\cite{MPW,KalninsMillerPost2013}.      

A distinctive feature of superintegrable systems is their nontrivial algebraic structure: while Liouville integrability is characterized by an abelian algebra of integrals in involution, in the maximally superintegrable case, the Poisson brackets between certain integrals no longer vanish. Consequently, one can introduce a \emph{symmetry algebra} associated with a system, which is typically a polynomial Poisson algebra, generated by the integrals of motion. This symmetry algebra is a powerful tool to characterize the system's dynamics and properties; in particular, it is  used for the classification of superintegrable systems~\cite{Kress2007}. Furthermore, in a maximally superintegrable case, one can, in principle, determine the system's behavior (\ie{}, classical trajectories or quantum spectra) purely algebraically, bypassing the need to integrate Hamilton's or Schr\"odinger's differential equations~\cite{MPW}.

The notions of integrability and superintegrability can be extended to \emph{discrete-time} dynamical systems, \ie{}, models whose dynamics are governed by difference equations. Formally, such systems can be described by a time-evolution map of the form $ t \mapsto t+h $, where~$ h>0 $ is called the \emph{discretization parameter} or the \emph{time step}. They have applications in biology~\cite{HeningSabharwal2025},  physics~\cite{Lee1983}, and naturally arise during the numerical analysis of differential equations. In the latter case, one faces a \emph{discretization problem}, \ie{}, finding a system of finite difference equations that converges to a given system of differential equations in the continuous limit $ h \to 0 $. Here, it is important to note that a discretization problem is not uniquely posed; moreover, different discretizations may have different properties while converging to the same ODE (or PDE) system in the continuous limit. In particular, a discretization of a continuous-time superintegrable system may be merely integrable (as for the caged oscillator model~\cite{Evans2008}, see~\cite{GubLat_sl2}), or nonintegrable at all (see~\eg{},~\cite{Petrera2011,Gubbiotti2021Nahm}). In this regard, the search for proper discretizations that retain the desired properties (such as integrability or superintegrability) remains an important research direction within the field; see, \eg{},~\cite{surisBook2003}.

Given the rich algebraic structure of superintegrable models, one might attempt to study the problem of superintegrable discretization on the algebraic level by examining the symmetry algebras of both the continuous model and its superintegrable discretization. Recently, in~\cite{DGL2026Calogero}, the symmetry algebras of the celebrated rational Calogero--Moser model~\cite{Calogero1971,Calogero1971erratum,wojciechowski1983} and its Nijhoff--Pang discretization~\cite{nijhoffTimediscretizedVersionCalogeroMoser1994,ujinoAdditionalConstantsMotion2008} were investigated. Specifically, it was shown that in this case discretization induces a deformation of the symmetry algebra of the continuous-time system with respect to the time step $h $, which plays the r\^ole of a deformation parameter. This result suggests an inverse problem, namely, whether the symmetry algebra formalism can be used to construct new superintegrable discretizations in an algorithmic way.

However, at the moment, the number of examples of this symmetry algebra deformation phenomenon is too limited to develop a systematic theory. In this regard, a natural candidate for further case study is the rational Ruijsenaars--Schneider model~\cite{RS1986}, which is commonly regarded as a relativistic generalization\footnote{Although not in the na\"ive sense of ``relativistic invariance'', see, \eg{}, the discussion in~\cite{Braden1997} on physical aspects of Ruijsenaars--Schneider models.} of the Calogero--Moser system,  at the same time sharing close similarities with the latter one on the level of algebraic structures~\cite{SurisWhy1996}.  The continuous-time model was proven to be not only integrable, but maximally superintegrable~\cite{AyadiFeher2010}. There also exists an integrable discretization of the Ruijsenaars--Schneider system~\cite{NijhoffRagnisco1996}; however, the superintegrability properties of this discretization, to the author's knowledge, have not yet been discussed in the literature. 
  
 In light of the above, this paper aims to extend the results of~\cite{ujinoAdditionalConstantsMotion2008,DGL2026Calogero} from the Calogero--Moser to the relativistic case. First, in~\cref{sec:bg}, we discuss the necessary preliminaries and known results on discrete-time systems, superintegrability, and the Ruijsenaars--Schneider model. Then, we present the following results: 
 \begin{itemize}
 	\item In~\cref{sec:SI-dRS}, we prove that the $N$-body discrete Ruijsenaars--Schneider system is superintegrable by constructing $N-1$ additional modified constants of motion of the corresponding system of difference equations. 
 	\item In~\cref{sec:Symalg-cont}, we obtain the complete structure of the symmetry algebra $\Alg^{(N)} $ associated with the continuous-time $N$-body Ruijsenaars--Schneider model. We illustrate our construction by presenting two- and three-body examples. 
 	\item In~\cref{sec:Symalg-discr}, we compute the symmetry algebra $\widetilde{\Alg}^{(N)}$ of the $N$-body discrete Ruijsenaars--Schneider system, and demonstrate that it is a deformation of the continuous-time symmetry algebra $\Alg^{(N)} $, with the time step $h$ serving as a deformation parameter, thereby extending an analogous result~\cite{DGL2026Calogero} from the Calogero--Moser model to the relativistic case. Similarly, we present examples for $N=2,3$.   
 \end{itemize}
 Finally, in~\cref{sec:concl}, we compare the obtained results with those for the Calogero--Moser case and discuss potential applications along with possible further developments.

\newpage 
\section{Background material and overview of existing results}
\label{sec:bg}

In this section, we review the fundamental notions from the theory of integrability and superintegrability for continuous-time~\cite{MPW,EscobarRuiz2026} and discrete-time dynamical systems~\cite{Bruschietal1991,Veselov1991,HietarintaJoshiNijhoff2016}, and recall the existing results on the rational Ruijsenaars--Schneider model~\cite{RS1986,arutyunovElementsClassicalQuantum2019} (hereinafter referred to as the \emph{RS model}), and its integrable discretization~\cite{NijhoffRagnisco1996}, see also~\cite[chapter 27]{surisBook2003}. 

We remark that this section is intended to establish notation and conventions, as well as to provide a brief overview of the context rather than an exhaustive discussion of the subjects. Therefore, for a comprehensive exposition of the topics, we invite the reader to consult the original references cited above. 

\subsection{Superintegrability of continuous- and discrete-time  dynamical systems}
\label{sec:SI-review}
\paragraph{Continuous-time case.}
We start by reviewing the notions for continuous-time dynamical systems in classical mechanics. Let us first recall the following:
\begin{definition}\label{def:SI}
Consider a Hamiltonian system with $N$ degrees of freedom, governed by a Hamiltonian function $H$. Suppose that there exists a set of $ k \geq N $ globally defined, functionally independent integrals of motion for the system. Then, we say that the system is:
\begin{itemize}
	\item \emph{(Liouville) integrable}, if $ k= N $ and if these $N$ integrals Poisson-commute; 
	\item \emph{superintegrable}, if it is integrable but $ k > N $. In particular: 
	\begin{itemize}
		\item \emph{minimally superintegrable}, if $ k= N+1 $, 
		\item \emph{maximally superintegrable} (hereinafter \emph{MS}), if $ k=2N-1 $. 
	\end{itemize}
\end{itemize}
\end{definition}
In the MS case, a remarkable property holds: every bounded trajectory is closed, and the motion is strictly periodic~\cite{nekhoroshev1972action}. Moreover, for such systems, one can, in principle, find a classical trajectory without calculus, by instead solving systems of algebraic equations (see~\cite{EscobarRuiz2026} for explicit examples of finding trajectories). 

From~\cref{def:SI}, it follows that for each $N$-degree-of-freedom MS system one can define a set:
\begin{equation}\label{eq:def:SNfi-intro}
	\genset^{(N)}_{\text{f.i.}} \coloneqq 
	\set{ F_1, \ldots, F_N , G_1, \ldots, G_{N-1} }
\end{equation}
of all $2N-1 $ functionally independent integrals of motion. Since by default we assume superintegrable systems to be integrable, there exists a subset of $N$ Poisson-commuting integrals of motion, such that $\{ F_i, F_j \} =0 $ for $i,j=1,\ldots,N $. The remaining Poisson brackets of the form $\{ F_i, G_m \}  $ and $ \{G_m, G_n \} $ (for $m,n=1,\ldots,N-1$) are, in general, nonzero. For a wide class of models,  one might express the right-hand side of these brackets as polynomials in integrals of motion. To this end, it may be necessary to extend the set~\eqref{eq:def:SNfi-intro} by adding $\ell \geq 0 $ functionally \emph{dependent} (yet linearly independent) additional integrals of motion~\cite{MPW}:
\begin{equation}\label{eq:def:SN-intro}
	\genset^{(N)} \coloneqq  
	\set{ F_1, \ldots, F_N , G_1, \ldots, G_{N-1}, I_1, \ldots, I_\ell }, 
\end{equation}  
Then, it might be possible to express the right-hand side of any nonzero relation $\{ F_i, G_m \}  $, $ \{G_m, G_n \} $, $\{ F_i, I_r \}  $, $ \{G_m, I_r \} $, $\{ I_r,  I_s \} $ for $i=1,\ldots,N $; $m,n =1,\ldots,N-1 $; $r,s =1,\ldots,\ell $ as a polynomial in the elements of~\eqref{eq:def:SN-intro}. The number $\ell $ depends on the model and on the number $N$: for example, it was rigorously shown in~\cite{EscobarRuiz2026} that for systems with $N=2$ degrees of freedom, it is always possible to obtain a symmetry algebra with $\ell =1 $. This is an upper bound: for some \mbox{$N=2$} models, $\ell =0 $ is sufficient as well; see, \eg{},~\cite[example II.1]{DGL2026Calogero}. There is also an analogous result for second-order systems with $N=3$ degrees of freedom; see~\mbox{\cite[theorem 14]{MPW}}.\footnote{Note that in this statement, the integrals of motion were assumed to be polynomials of degree two in momenta.}

Together with a generating set $\genset^{(N)} $, the Poisson structure relations define a \emph{symmetry algebra}~$\Alg^{(N)} $ associated with an $N$-degree-of-freedom MS system. This algebra encodes information about the dynamics and symmetries of the model under investigation; in particular, symmetry algebras are used for the classification of superintegrable systems~\cite{Kress2007}. While~$\Alg^{(N)} $ is typically a polynomial Poisson algebra, in simple cases -- such as the oscillator models or the Kepler problem -- it may not contain nonlinear terms, thereby becoming a Lie algebra; see, \eg{},~\cite{DGLOsc2025}.

For illustrative examples of explicit symmetry algebra computations in two degrees of freedom, we invite the reader to consult~\cite[section 3]{EscobarRuiz2026}. Let us also mention the works \cite{DGLOsc2025, DGL2026Calogero}, in which the symmetry algebras of the various $N$-dimensional oscillator models and the $N$-body Calogero--Moser system were computed, respectively.%

\paragraph{Discrete-time case.} 
Let us now consider  discrete time $ t \in h \Z $, where $ h>0$ is a fixed \emph{discretization parameter} or \emph{time step}, and a dynamical variable  $ \vec y (t) = 
( y_1(t), \ldots, y_d (t) ) \in \R^d $.  It is convenient to introduce the following compact notation:
\begin{equation}
	y_i \equiv y_i (t), \qquad
	\bar y_i \coloneqq y_i ( t+ h), \qquad
	\ubar{y}_i \coloneqq y_i (t-h), \qquad
	i=1,\ldots, d.  
\end{equation} 
Then, an expression of the form:
\begin{equation}\label{eq:1sr-ord-OdiffE}
	\vecbar{y} = \vec{f} ( \vec y) 
\end{equation}
is called a system of first-order ordinary (finite) \emph{difference equations (O$\Delta$Es)} for an unknown vector-valued sequence $\set{ \vec y (t) }_{t \in h \Z } \subset \R^d $. Fixing an initial condition $ \vec y ( t_0) \eqqcolon \vec y_0 $,~\cref{eq:1sr-ord-OdiffE} is equivalent~\cite{GJTV_class} to the iteration of the finite-translation map  $ \vec \Phi \colon \vec y \mapsto \vecbar{y}  $.      

In this setting, the discrete-time analog of a Hamiltonian system is the theory of \emph{symplectic maps} (or correspondences) of the form: 
\begin{equation}
\begin{aligned}
	\vec \Phi \colon D & \to D
	\\
	\vec y & \mapsto \vec \Phi ( \vec y) = \vecbar{y}, 
\end{aligned}	
\end{equation} 
where $D $ is an open subset of a symplectic manifold with symplectic form $\omega $. The map $\vec \Phi $ preserves a symplectic structure, \ie{}, $\vec \Phi^\ast  \omega = \omega  $. This condition provides an analog of continuous-time evolution through a Hamiltonian flow, because time evolution is required to be a canonical transformation. Therefore, we can define our main object of study as follows:
\begin{definition}
	Let $ \vec \Phi \colon D  \to D $ be a symplectic map, and fix an initial condition $ \vec y ( t_0) \eqqcolon \vec y_0 $. Then, we call an iteration $ \set{ \vec \Phi^n (\vec y_0) }_{n \in \Z} $ a \emph{discrete-time dynamical system}. 
\end{definition}
 
Similarly to the continuous setting, we say that a scalar function $J (\vec y) $ is an \emph{integral of motion} for the map $\vec \Phi $ if it remains constant under the action of $\vec \Phi $, \ie{}, \mbox{$ \bar{J} (\vec y ) \equiv J (\vecbar{y} ) = J (\vec y) $}. Then, in general, the theory of integrability for symplectic maps is developed in complete analogy with the continuous case, with one crucial difference: in a discrete setting, there is no direct analog of a Hamiltonian function. A canonical transformation $\vec \Phi $ has an associated generating function, but unlike the Hamiltonian, it is not, in general, an integral of motion. Consequently, discrete integrable systems appear to be much rarer and possess a more rigid structure compared to their continuous counterparts.

Let us now proceed to the superintegrability of symplectic maps. Using the definitions above, one can extend~\cref{def:SI} to discrete-time systems:  

\begin{definition}\label{def:discrete-SI}
Let $ \vec \Phi \colon D  \to D $ be a symplectic map. 
 Suppose that there exists a set of $ k \geq N $ globally defined, functionally independent integrals of motion for the map $\vec \Phi $. Then, we say that the map is:
\begin{itemize}
	\item \emph{(Liouville) integrable}, if $ k= N $ and if these integrals Poisson-commute; 
	\item \emph{superintegrable}, if it is integrable but $ k > N $. In particular: 
	\begin{itemize}
		\item \emph{minimally superintegrable}, if $ k= N+1 $, 
		\item \emph{maximally superintegrable} (\emph{MS}), if $ k=2N-1 $. 
	\end{itemize}
\end{itemize}
\end{definition}

To apply algebraic methods to study the dynamics of MS symplectic maps, we can also extend the notion of a symmetry algebra to the discrete-time case. The primary difference is that in the discrete case, the symmetry algebra no longer needs to possess a central element, which stems from the absence of the Hamiltonian function; see~\cite{DGL2026Calogero} and the related discussion above.

Finally, let us comment on the discretization problem. If, in a continuous limit $h \to 0 $, possibly up to a parameter rescaling, the system of difference equations~\eqref{eq:1sr-ord-OdiffE} converges to a system of first-order differential equations: 
\begin{equation}\label{eq:1st-ord-ODE}
	\vec{\dot Y} = \vec{ g } ( \vec Y ),
\end{equation}
we say that~\eqref{eq:1sr-ord-OdiffE} is a \emph{discretization} of~\eqref{eq:1st-ord-ODE}.  

Note that, in general, the discretization problem is ill-posed, in the sense that it does not admit a unique solution. Moreover, different discretizations can have different properties, both from analytical and numerical points of view (see, for instance, the discussion in~\cite{LeviMartinaWinternitz2015}). Consequently, finding discretizations that preserve specific desirable properties remains an important problem within the field, see \eg{},~\cite{surisBook2003}. 

In particular, we are interested in preserving \emph{superintegrability} properties when discretizing a given continuous system. Indeed, a discretization of a continuous superintegrable system may be merely integrable, or lose integrability entirely (see, \eg{}, the caged oscillator case discussed in~\cite{GubLat_sl2}).  Therefore, if both the system of ODEs~\eqref{eq:1st-ord-ODE} and the system of O$\Delta$Es~\eqref{eq:1sr-ord-OdiffE}  are (super)integrable, we say that the latter is a~\mbox{\emph{(super)integrable~discretization}}.

\subsection{Rational Ruijsenaars--Schneider model and its integrable discretization}
\label{sec:bg-RS-review}
\paragraph{Continuous-time case.} We consider a $2N$-dimensional symplectic manifold with Darboux coordinates $\vec q = ( q_1, \ldots, q_N) $ and $ \vec p = ( p_1, \ldots, p_N) $, satisfying the canonical Poisson brackets:
\begin{equation}
	\{ q_i, p_j \} = \delta_{i,j}, \qquad
	\{ q_i, q_j \} = \{ p_i, p_j \} =0, \qquad 
	1 \leq i,j \leq N.  
\end{equation} 
The rational RS model~\cite{RS1986} introduces the time-translation generator $\calH  $, the space-translation generator $ \calP $, and the boost generator $\calB $ with the following dynamical realization:
\begin{subequations}\label{eq:Poincare-generators-realization}
\begin{align}
	\calH & = m c^2 \sum_{k=1}^N \cosh (  \frac{p_k}{mc} ) 
	\prod_{ j \ne k } \sqrt{ 1+ \frac{\chi^2}{ m^2 c^2 \,  (q_k - q_j )^2}  } , 
	\label{eq:Poincare-generators-realization-H}
	\\ 
	\calP & = m c \sum_{k=1}^N \sinh ( \frac{p_k}{mc} ) 
	\prod_{ j \ne k } \sqrt{ 1+ \frac{\chi^2}{ m^2 c^2\, (q_k - q_j )^2}  }, 
	\label{eq:Poincare-generators-realization-P}
	\\
	\calB & = - m \sum_{k=1}^N q_k, 
\end{align}	
\end{subequations}
where $m$ is the particle mass, $c $ represents the speed of light, and we assume the parameter $\chi $ to be real. The generators~\eqref{eq:Poincare-generators-realization} obey the $1+1 $ Poincar\'e algebra:
\begin{equation}
	\{ \calH, \calP \} =0, \qquad
	\{ \calH, \calB \} = \calP, \qquad 
	\{ \calP, \calB \} = c^{-2} \calH. 
\end{equation}
Hereinafter, let us set $ m = c =1 $. Note that~\eqref{eq:Poincare-generators-realization} are not polynomial in momenta. In this regard, it is convenient to introduce the quantity: 
\begin{equation}
	b_k \coloneqq e^{p_k} \prod_{j \ne k } \sqrt{ 1+ \frac{\chi^2}{ ( q_k - q_j )^2 }   }, 
\end{equation}
such that:
\begin{equation}
	\{ q_i, q_j \} =0, \qquad
	\{ q_i, b_j \} = b_i \delta_{i,j}, \qquad
	\{ b_i, b_j \} =  \pi_{i,j} b_i b_j, 
	\qquad 
	1 \leq i,j \leq N, 
\end{equation}
where:
\begin{equation}
	\pi_{i,j} \coloneqq  
	\frac{1}{q_j - q_i + \imath \chi  }
	-\frac{1}{q_i - q_j + \imath \chi }
	+ (1- \delta_{i,j}) \frac{2 }{ q_i - q_j  }. 
\end{equation}

To discuss integrability, we define the following matrix $L$ by its matrix elements: 
\begin{equation}	\label{eq:Lax-matrix-elem-def} 
	L_{i,j} = \frac{ \imath \chi b_j }{ q_i - q_j + \imath \chi  }. 
\end{equation}
Observe that now we can rewrite the time-translation generator~\eqref{eq:Poincare-generators-realization-H} and the space-translation generator~\eqref{eq:Poincare-generators-realization-P} as follows:
\begin{equation}\label{eq:HandPviaHplusAndHminus}
	\calH  = \tfrac{1}{2} ( H_+ + H_- ), 
	\qquad
	\calP  = \tfrac{1}{2} ( H_+ - H_- ),
\end{equation} 
 where:  
\begin{equation}\label{eq:Hplus-Hminus-def}
	H_+ \coloneqq  \trace( L ), 
	\qquad
	H_- \coloneqq  \trace( L^{-1} ).   
\end{equation}

The equations of motion for the RS Hamiltonians~\eqref{eq:Hplus-Hminus-def} are equivalent to the following isospectral evolution equations:
\begin{subequations}\label{eq:Lax-eqs}
\begin{align}
\dv{L}{t} &= [ M^\pm, L ], 
\\ 	
\dv{X}{t} &  = [ M^\pm, X] + \varphi (L), 
\end{align}	
\end{subequations}
where $X \coloneqq \diag(q_1, \ldots, q_N) $, $ \varphi (L) \colon \gl_N \to \gl_N $ is a conjugation-covariant function~\cite{surisBook2003}, and the companion matrix is given, for example, as follows:  
\begin{equation}
	M_{i,j}^+  =  - \delta_{i,j} \left(  \frac{\imath}{ \chi } b_i  
	+ \sum_{k \ne i } \frac{ \imath \chi  }{ (q_i - q_k  ) ( q_i - q_k + \imath \chi  )  } b_k  
	\right) 
	+ (1 - \delta_{i,j} )  \frac{b_j}{ (q_i -q_j ) }. 
\end{equation}
The spectral invariants of the Lax matrix $L  $ provide a set of $N$ Poisson-commuting functionally independent integrals of motion $ \set{ F_1, \ldots, F_N } $, where:
\begin{equation}\label{eq:Fk-def}
	F_k \coloneqq \trace ( L^k ), \qquad k \in \Z.  
\end{equation}
Therefore, the rational RS system is integrable. Let us select the simplest member of the hierarchy as the Hamiltonian of the model~\cite{SurisWhy1996}:
\begin{equation}\label{eq:Hplus-choice}
	H_+ = F_1 = \sum_{k=1}^N b_k. 
\end{equation}
 Its flow is given by the integral curves of the following equations of motion: 
\begin{equation}\label{eq:eoms-F1} 
	\dot q_i = \{ q_i, F_1 \} = b_i, 
	\qquad
	\dot b_i = \{ b_i, F_1 \} =  \sum_{j \ne i }^N \pi_{i,j } b_i b_j.
\end{equation}
\begin{remark} 
Let us comment on the nonrelativistic limit of the RS system. 
Taking the limit of~\eqref{eq:Poincare-generators-realization} as $c \to \infty $, we obtain:
\begin{equation}
	\calH = N m c^2  +  H_{\CM} + \mathcal{O} \left( \frac{1}{c^2} \right ), 
	\qquad
	\calP =  P_{\CM} + \mathcal{O} \left( \frac{1}{c^2} \right),
\end{equation}
where $H_\CM  $ is the Hamiltonian of the Calogero--Moser model~\cite{Calogero1971,Calogero1971erratum}, and \mbox{$ P_\CM \coloneqq \sum_{k=1}^N p_k $}. Our choice~\eqref{eq:Hplus-choice}, in the physical units, reads as:
\begin{equation}
	H_+ = m c^2 \sum_{k=1}^N \exp \left( \frac{p_k}{mc}  \right)  \prod_{j \ne k } \sqrt{ 1+ \frac{\chi^2}{ m^2 c^2 \,  (q_k - q_j )^2 }   },
\end{equation}
and corresponds to the following nonrelativistic limit:
\begin{equation}
	H_+ = N m c^2  + c P_{\CM} + H_{\CM} + \mathcal{O} \left( \frac{1}{c } \right),
\end{equation}
which was expected, given~\cref{eq:HandPviaHplusAndHminus}.  
\end{remark}

The superintegrability of the rational RS model was established in~\cite{AyadiFeher2010}, see also~\cite{AFG2012}. Let us review the construction of the additional integrals of motion. First, using again the  diagonal matrix~\mbox{$ X $}, we define the auxiliary dynamical functions:
\begin{equation}
	J_k \coloneqq \trace( X L^k ), \qquad k \in \Z,  
\end{equation}
which obey the following structure relations:
\begin{equation}\label{eq:JFandJJ}
	\{ J_k, F_j \} = j F_{k+j}, 
	\qquad
	\{ J_k, J_j \} = (j-k) J_{k+j}. 
\end{equation}
Consequently, one can construct the quantities of the form: 
\begin{equation}\label{eq:Ks-new-def}
	K^{(a)}_{ m,n } \coloneqq J_m F_{a+n} - J_n F_{a+m}. 
\end{equation}
The quantities~\eqref{eq:Ks-new-def} constitute a relativistic generalization of the Wojciechowski-type constants of motion for the rational Calogero--Moser model~\cite{wojciechowski1983}. 
One can prove that each~$ K^{(a)}_{ m,n }  $ Poisson-commutes with a given $F_a $:
\begin{equation}
	\{ K^{(a)}_{ m,n }, F_a \} = \{ J_m F_{a+n}, F_a \} 
	- \{ J_n F_{a+m}, F_a\} 
	= F_{a+n} \{ J_m , F_a \} - F_{a+m} \{ J_n, F_a  \}
	=0.    
\end{equation}
Moreover, for a fixed $ a$, one can select a subset of $N-1 $ functionally independent integrals of motion among them, ensuring the superintegrability of the model whose dynamics are governed by $F_a $.  

We will focus on the model governed by $F_1 $; therefore, let us fix $a=1 $ in~\cref{eq:Ks-new-def}. In what follows, we will omit the upper index for $a=1 $:
\begin{equation}\label{eq:Kmn-cont}
	K_{m,n} \equiv K_{m,n}^{(1)} 
	= J_m F_{n+1} - J_n F_{m+1}. 
\end{equation}   
Note the antisymmetry $K_{m,n} = - K_{n,m} $. 
Hence, the integrals of motion of the form~\eqref{eq:Ks-new-def} for $F_1 $ can be organized into the following matrix:
\begin{equation}\label{eq:K-matrix}
\begin{bmatrix}
 0 & -K_{2,1} & \cdots & -K_{N, 1 }
 \\ 
 K_{2,1} & 0 & \cdots &  -K_{N , 2} 
 \\
 \vdots & & \ddots & \vdots 
 \\ 
 K_{N,1} & K_{N,2} & \cdots & 0 
\end{bmatrix}.
\end{equation}
The subset of $N-1 $ functionally independent integrals can be chosen, for example, by selecting the first column of~\eqref{eq:K-matrix}. That is, the set:    
\begin{equation}\label{eq:SN-fi}
	\genset^{(N)}_{\text{f.i.}} \coloneqq 
	\set{ F_1, \ldots, F_N, G_1, \ldots, G_{N-1}},
	\qquad
	G_m \coloneqq K_{m+1,1}  
\end{equation}
consists of a total of $2N-1 $ functionally independent constants of motion, ensuring the maximal superintegrability of the model governed by $F_1 = H_+ $.  

Other integrals from~\eqref{eq:K-matrix} are functionally dependent on those in the set $\genset^{(N)}_{\text{f.i.}} $: indeed, using the definitions~\eqref{eq:Fk-def} and~\eqref{eq:Ks-new-def}-\eqref{eq:Kmn-cont}, one can construct the functional relation:  
\begin{equation}\label{eq:funrel-cont}
	F_{j+1} K_{n,i} - F_{i+1} K_{n,j} + F_{n+1} K_{i,j} =0, 
	\qquad
	i,j,n \in \Z. 
\end{equation}
Fixing $ i=1 $ in~\cref{eq:funrel-cont} and using the antisymmetry of $K_{i,j} $, we obtain:
\begin{equation}\label{eq:funrel-cont-G}
	F_{j+1} G_{n-1} - F_{2} K_{n,j} - F_{n+1} G_{j-1} =0.  
\end{equation}
Here, \cref{eq:funrel-cont-G} relates $K_{n,j} $ for $2 \leq j < n \leq N -1  $, to the elements of the set $ \genset^{(N)}_{\text{f.i.}} $. Moreover, we can extend this argument to the case $ n= N $ because $F_{N+1} $ can be expressed uniquely in terms of $\set{ F_1, \ldots, F_N } $ using the Cayley--Hamilton theorem (we will discuss this in more detail below when computing the symmetry algebras).

Finally, from~\eqref{eq:K-matrix} and~\eqref{eq:funrel-cont-G}, it follows that for each fixed $N$ we need:
\begin{equation}\label{eq:funrel-counting}
	\frac{1}{2} N (N-1) - (N-1) = 
	\frac{1}{2} (N-1) (N-2) 
\end{equation} 
identities of the form~\eqref{eq:funrel-cont-G} to relate functionally dependent integrals from~\eqref{eq:K-matrix} to the elements of~$\genset^{(N)}_{\text{f.i.}} $. 
\paragraph{Discrete-time case.} 
Let us now proceed to the discretization of the rational RS system. 
The integrable discretization of the eqs.~\eqref{eq:eoms-F1} is given~\cite{NijhoffRagnisco1996,surisBook2003} by the symplectic map $\vec \Phi^+ \colon (q, b) \mapsto (\bar q, \bar b)   $, which, using the conventions of~\cref{sec:SI-review}, can be explicitly written in canonical form as follows:
\begin{subequations}\label{eq:Phi_plus_eoms}
\begin{align}
	b_k & = - \frac{\imath \chi }{h} 
	\frac{\prod_{i=1}^N ( q_k - \bar q_i) \prod_{i \ne k } (q_k - q_i + \imath \chi)  }{ \prod_{i=1}^N ( q_k - \bar q_i + \imath \chi  ) 
	\prod_{ i \ne k  } ( q_k - q_i)  }, 
	\\
	\bar b_k & = - \frac{\imath \chi }{h} 
	\frac{\prod_{i=1}^N ( \bar q_k -  q_i) 
	\prod_{i \ne k } (\bar q_k - \bar  q_i - \imath \chi)  }
	{ \prod_{i=1}^N ( \bar q_k -  q_i - \imath \chi  ) 
	\prod_{ i \ne k  } ( \bar q_k - \bar q_i)  }.  
\end{align}
\end{subequations}
Similarly, there exists a discretization corresponding to the flows of $H_- = F_{-1} $, but it is equivalent to $\vec \Phi^+ $ under the change $ \chi \mapsto - \chi  $ and the interchange $ \bar{q}_k \leftrightarrow \ubar{q}_k $, see~\cite[\S 27.1]{surisBook2003}. 

The discrete dynamical system~\eqref{eq:Phi_plus_eoms} admits a Lax representation. That is, one can discretize eqs.~\eqref{eq:Lax-eqs} as follows:
\begin{subequations}\label{eq:dLax-eqs}
\begin{align} 
	\bar{L} \M &= \M L,  \label{eq:dLax1}
	\\  
	\bar{X} \M &= \M X + 
	 h \M L \left( \1 - \imath h  \chi^{-1} L \right)^{-1},
	 \label{eq:dLax2} 
\end{align}
\end{subequations} 
where $L $ is given by the same matrix elements~\eqref{eq:Lax-matrix-elem-def}, and the modified companion matrix takes the form: 
\begin{equation}
	\M_{i,j} \coloneqq \sum_{i,j=1}^N \frac{h b_j}{ \bar q_i - q_j  }. 
\end{equation} 
It is easy to verify that the quantities $F_k $, defined by~\cref{eq:Fk-def}, are conserved under the difference equations of motion~\eqref{eq:Phi_plus_eoms}. Indeed, using the cyclicity of the trace, 
 one obtains, from~\eqref{eq:dLax1}:
\begin{equation}\label{eq:Fk-discr-conservation}
	\bar F_k = \trace \bigl[ \bar L^k \bigl] 
	= \trace \left[ \M L^k \M^{-1}  \right]  
	= \trace \bigl[ L^k  \bigl] = F_k.  
\end{equation}

Therefore, similarly to the continuous case, the set $\set{ F_1, \ldots, F_N } $ consists of $N$ Poisson-commuting functionally independent integrals of motion for the map~$\vec \Phi^+ $, ensuring the integrability of the discrete RS model.  

However, the additional integrals of motion~\eqref{eq:Ks-new-def} for the continuous system are no longer conserved in the discrete case. In this regard, we will derive their modified versions and discuss the superintegrability properties of the discrete model in the following section.   

\newpage 

\section{Superintegrability of discrete-time Ruijsenaars--Schneider model}
\label{sec:SI-dRS}
In~\cref{sec:bg}, we discussed the integrability and superintegrability of the continuous RS model, and the integrability of the discrete-time RS model. Therefore, to proceed to the computation of symmetry algebras, it only remains to show that the discretization is also superintegrable, like its continuous counterpart.   
 
As mentioned above, one can verify by direct computation that the additional integrals of motion $K_{m,n} $ for the continuous-time RS model, defined by~\cref{eq:Ks-new-def}, are no longer conserved in the discrete-time case. Nevertheless, the discretization~\eqref{eq:Phi_plus_eoms} is also superintegrable, and we can show this by direct construction of the modified additional Wojciechowski-type constants of motion. 

In fact, the construction of the modified invariants is very similar to the rational Calogero--Moser case: see how~\cite{ujinoAdditionalConstantsMotion2008} extends the result of~\cite{wojciechowski1983} to the discrete-time system. However, to the author's knowledge, the proof for the RS system is not yet written; therefore, in this section we fill this gap and provide this proof for consistency. This is the content of the following:
\begin{prop}\label{prop:discrete-additional-constants}
The modified quantities: 
\begin{equation}\label{eq:discrete-additional-constants-def}
\K_{m,n} \coloneqq \trace\left[ X \left( \1 - \imath h \chi^{-1} L  \right) L^m     \right] 
   \trace\left( L^{n+1} \right) 
   - \trace\left[ X \left( \1 - \imath h \chi^{-1} L  \right) L^n  \right]
  \trace\left( L^{m+1} \right), 
\end{equation} 
where $\1 $ is an $N \times N $ identity matrix, 
are conserved under the map $\vec  \Phi^+ \colon (q,b) \mapsto (\bar q, \bar b) $, given by the difference eqs.~\eqref{eq:Phi_plus_eoms}. Moreover, we have $\lim_{ h \to 0 } \K_{m,n} = K_{m,n} $ in the continuous limit. Additionally, the modified integrals of motion obey the following functional relation: 
\begin{equation}\label{eq:funrel-discr}
	F_{j+1} \K_{n,i} - F_{i+1} \K_{n,j} + F_{n+1} \K_{i,j} =0, 
	\qquad
	i,j,n \in \Z.  
\end{equation}
\end{prop}
\begin{proof} 
The continuous limit can be verified simply by evaluating~\cref{eq:discrete-additional-constants-def} at~\mbox{$h=0$}. Hence, the quantities $\K_{m,n} $ reduce to the integrals of motion $K_{m,n} $ of the continuous model, see~\cref{eq:Kmn-cont,eq:Ks-new-def}. The functional relations again follow from the definitions~\eqref{eq:discrete-additional-constants-def} and~\eqref{eq:Fk-def}.  

Let us now show that the modified quantities $\K_{m,n} $ are conserved under the system~\eqref{eq:Phi_plus_eoms}. The proof is performed by following the exact same steps as for the Calogero--Moser case.  Applying  the map $\vec \Phi^+ $ to~\eqref{eq:discrete-additional-constants-def}, we obtain: 
\begin{equation}\label{eq:proof:discrete-additional-constants-interm1}
	\bar{ \K}_{m,n} = \trace\left[ \bar X \left( \1 - \imath h \chi^{-1} \bar L  \right) \bar{L}^m   \right] 
   \trace\left( \bar{L}^{n+1} \right) 
   - \trace\left[ \bar X \left( \1 - \imath h \chi^{-1} \bar L  \right) \bar{L}^n  \right]
  \trace\left( \bar{L}^{m+1} \right).  
\end{equation}
Note that, combining~\cref{eq:dLax1,eq:dLax2}, we have: 
\begin{equation}\label{eq:proof:discrete-additional-constants-interm2}
\bar X \left( \1 - \imath h \chi^{-1} \bar{L} \right) 
= 
\M \left[ X  \left( \1 - \imath h \chi^{-1} L \right) + h L  \right] \M^{-1}.
\end{equation}
Substituting~\cref{eq:proof:discrete-additional-constants-interm2} into~\cref{eq:proof:discrete-additional-constants-interm1}, with the use of~\cref{eq:Fk-discr-conservation} and the trace properties, we arrive at: 
\begin{equation}
\begin{split}
\bar{ \K}_{m,n} &= 
\trace\left[  X  \left( \1 - \imath h \chi^{-1} L \right)    L^m
 + h L^{m+1}     \right] 
   \trace\left( {L}^{n+1} \right) 
   - \trace\left[   X  \left( \1 - \imath h \chi^{-1} L \right) L^n + h L^{n+1}     \right]
  \trace\left( {L}^{m+1} \right)  
  \\
  & = 
  \trace\left[  X  \left( \1 - \imath h \chi^{-1} L \right)    
  L^m       \right] 
   \trace\left( {L}^{n+1} \right) 
   - \trace\left[   X  \left( \1 - \imath h \chi^{-1} L \right) L^n     \right]
  \trace\left( {L}^{m+1} \right)   
  \\
  & = \K_{m,n}. 
\end{split}	
\end{equation}
\end{proof}
The modified integrals of motion provide superintegrability for the discrete RS model. That is:
\begin{prop}\label{prop:discrete-system-is-MS}
The discrete RS system $\vec \Phi^+ \colon (q, b) \mapsto ( \bar q, \bar b )  $ with the equations of motion~\eqref{eq:Phi_plus_eoms} is maximally superintegrable. 
\end{prop}
\begin{proof}
Similarly to the continuous case, one can choose a set of $2N-1 $ integrals of motion as follows:
\begin{equation}\label{eq:discr-SN-fi}
	\widetilde \genset_{\text{f.i.}}^{(N)} \coloneqq 
	\Set{ F_1, \ldots, F_N, \G_1, \ldots, \G_{N-1} },
\end{equation}
where:
\begin{equation}
	\G_m \coloneqq \K_{m+1,1}
	=
	\trace\left[ X \left( \1 - \imath h \chi^{-1} L  \right) L^{m+1}     \right] 
   \trace\left( L^{2} \right) 
   - \trace\left[ X \left( \1 - \imath h \chi^{-1} L  \right) L  \right]
  \trace\left( L^{m+2} \right).   
\end{equation}
The functional independence of the elements of $\widetilde \genset_{\text{f.i.}}^{(N)} $ clearly follows from the continuous-time case. 
\end{proof}

\begin{remark}\label{rmk:RSvsCM-discrete-via-gen-flows}
Note that one can rewrite the discrete constants of motion~\eqref{eq:discrete-additional-constants-def} via the continuous flows~\eqref{eq:Ks-new-def} as follows:
\begin{equation}\label{eq:Kdiscr-via-flows}
	\K_{m,n} =  K_{m,n} - \imath h \chi^{-1} K_{m+1, n+1}^{(0)}. 
\end{equation}
Formally, the quantities  $  K_{i,j}^{(0)} $ correspond to the ``integrals of motion'' for the trivial flow generated by~\mbox{$ F_0 = N = \mathrm{const} $}. While for the RS model alone, this construction seems trivial, since any dynamical variable trivially Poisson-commutes with a constant, we remark that \cref{eq:Kdiscr-via-flows} reveals the same pattern as observed also in the Calogero--Moser case~\cite[eq. (III.12)]{DGL2026Calogero}. Namely,  for both models, the additional integrals of motion in the discrete setting take the form:   
\begin{equation}
	\K_{m,n}^{(a)} = K_{m,n}^{(a)} - \varkappa( h) K_{m+1,n+1}^{(a-1)}, 
\end{equation}
where $ \varkappa( h) $ is a function of the discretization parameter. For the Calogero--Moser model, we considered $a=2 $ since it corresponds to the flow of the standard quadratic (in $\vec p $) Hamiltonian function in Newtonian classical mechanics. This observation shows that in both cases, discretization adds the additional term, corresponding to a ``step backward'' in the hierarchy of~$K^{(a)}_{m,n} $. 
\end{remark}

\newpage

\section{Symmetry algebras and their deformations}
In this section, we compute the symmetry algebras of the RS models. First, in~\cref{sec:Symalg-cont}, we obtain the structure relations in the continuous-time case.
Then, in~\cref{sec:Symalg-discr}, we extend this construction to the discrete-time setting, using the additional discrete integrals of motion obtained in~\cref{sec:SI-dRS} above. 
By comparing the obtained symmetry algebras, we show that the symmetry algebra of the discrete model constitutes a deformation of the symmetry algebra in the continuous-time case, with the time step $h $ serving as a deformation parameter. We also illustrate our construction with two- and three-body examples.

\subsection{Symmetry algebra in continuous-time case}
\label{sec:Symalg-cont}
Let us start from the simplest possible case, \ie{}, the $N=2$-body continuous-time RS model. We compute the symmetry algebra as follows: 
\begin{example}[$N=2$]\label{ex:N=2-cont}
In our conventions, the evolution of the continuous-time dynamical system~\eqref{eq:eoms-F1} in the two-body case is governed by:
\begin{equation}\label{eq:ex:N=2-cont-Hplus-explicit-expression}
	H_+ = F_1 = b_1 + b_2. 
\end{equation}
The superintegrability of the model is provided by the set~\eqref{eq:discr-SN-fi} of a total of~$3$ integrals of motion:
\begin{equation}\label{eq:ex:N=2-cont-S2}
	\genset^{(2)} = \genset^{(2)}_{\text{f.i.}} = \set{ F_1, F_2, K_{2,1} }.
\end{equation}
Explicitly, these integrals can be rewritten as:
\begin{equation}\label{eq:ex:N=2-cont-explicit-expressions}
	F_2 = b_1^2 + b_2^2 + \frac{2 \chi^2 b_1 b_2 }{ (q_1 - q_2)^2 + \chi^2 }  , 
	\qquad
	K_{2,1} = \frac{1}{2} ( F_1^2 - F_2 ) 
	\bigl( F_1 
	 (  \vec{q} \vdot   \vec{b} ) - (q_1 + q_2) F_2 
	\bigr).  
\end{equation}
The computation of the Poisson brackets leads to the following structure relations: 
\begin{equation}\label{eq:ex:N=2-cont-structure-relations}
	\{ F_1, F_2 \} =0, \qquad 
	\{ F_1, K_{2,1} \} =0, \qquad 
	\{ F_2, K_{2,1} \} = \frac{1}{2} ( F_1^2 - F_2 )^2 ( F_1^2 - 2 F_2 ).
\end{equation} 
Therefore, we obtain that the symmetry algebra $ \Alg^{(2)}$ of the two-body continuous-time RS model is a polynomial Poisson algebra of degree six with $ H_+ = F_1  $ being a central element.  

Observe that to express the right-hand side of~\eqref{eq:ex:N=2-cont-structure-relations} purely in terms of the elements of~\eqref{eq:ex:N=2-cont-S2}, we do not need to introduce any additional functionally dependent integrals, \ie{}, we have~\mbox{$\ell =0 $} in~\eqref{eq:def:SN-intro} and the set of functionally independent integrals of motion $\genset^{(2)}_{ \text{f.i.} } $ coincides with the generating set $\genset^{(2)} $ of the symmetry algebra (see~\cref{eq:def:SN-intro}). Note that this is an exceptional feature of the two-body case: in what follows, we will show that this is no longer true for model with a higher number of bodies.    
\end{example} 
Let us now proceed to the arbitrary $ N$-body case. First, we compute the Poisson bracket between $F_k $ and $K_{m,n} $. This is the content of the following statement: 
\begin{prop}\label{prop:cont-formal-rels}
The integrals of motion of the continuous-time RS model obey the following relations:  
\begin{subequations}\label{eq:prop:cont-formal-rels}
\begin{align}
 \{ F_k, F_m \} & = 0, 
\label{eq:prop:cont-formal-rels-FF}
 \\ 
	\{ F_k, K_{m,n} \} &= k \left( F_{m+1} F_{k+n} - F_{n+1} F_{k+m} \right),\label{eq:prop:cont-formal-rels-FK}
	\\
	\begin{split}
		\{ K_{i,j}, K_{m,n} \} &= 
	 F_{i+1} \bigl( (m+1) K_{n, j+m} + (n+1) K_{j+n, m}  \bigr) 
	 + F_{j+1} \bigl( (m+1) K_{i+m,n} +(n+1) K_{m, i+n}  \bigr) \\ 
	 & \quad + F_{m+1} \bigl( (i+1) K_{i+n, j} + (j+1) K_{i, j+n}  \bigr)
	 + F_{n+1} \bigl( (i+1) K_{j, i+m} + (j+1) K_{j+m,i} \bigr). 
	 \label{eq:prop:cont-formal-rels-KK}
	\end{split}
\end{align}	
\end{subequations}
\end{prop}
Note that the relations~\eqref{eq:prop:cont-formal-rels} alone do not lead to a finitely-generated algebraic structure. Indeed, the right-hand side of~\cref{eq:prop:cont-formal-rels-FK,eq:prop:cont-formal-rels-KK} returns Virasoro-like relations, where the resulting index is increased as a sum of the indices from the left-hand side. Since our goal here is to obtain a finitely-generated polynomial Poisson algebra, as discussed in the introductory~\cref{sec:SI-review}, we will need to perform additional algebraic constructions, provided later in this section. For now, let us prove the statement: 
\begin{proof}[\cref{prop:cont-formal-rels}] 
The relation~\eqref{eq:prop:cont-formal-rels-FK} can be derived simply by using the definition~\eqref{eq:Kmn-cont} together with the properties of the Poisson bracket (bilinearity and Leibniz rule). To derive~\cref{eq:prop:cont-formal-rels-KK}, in turn, one can introduce an object:
\begin{equation}
	R_{m,n} \coloneqq J_m F_{n+1}, 
\end{equation} 
such that the integrals $K_{m,n} $ are given by its antisymmetrization with respect to the indices: 
\begin{equation}
	K_{m,n} = R_{m,n} - R_{n,m} \eqqcolon R_{[n,m]}. 
\end{equation}
Evaluating the following auxiliary bracket:
\begin{equation}
	\{ R_{i,j}, R_{k,l} \} 
	= \{ J_i F_{j+1}, \, J_k F_{l+1} \} = 
	- (j+1) F_{l+1} R_{i, j+k}
	+ F_{j+1} \bigl( (k-i) R_{i+k, l} + (l+1) R_{k, i+l} \bigr), 
\end{equation}
and performing an antisymmetrization, we arrive at~\cref{eq:prop:cont-formal-rels-KK}. For more details, see the proof of an analogous statement for the Calogero--Moser model,~\cite[proposition II.3]{DGL2026Calogero}. 
\end{proof}
Now, let us apply~\cref{prop:cont-formal-rels} to the symmetry algebra computation.
First of all, we consider the relations between the elements of the set $\genset^{(N)}_{\text{f.i.}} $~\eqref{eq:SN-fi}. It follows from the trace properties and the Cayley--Hamilton theorem that to arrive at the finitely-generated Poisson algebra of integrals of motion, one can extend the set $\genset^{(N)}_{\text{f.i.}} $ by adding all the remaining integrals from~\eqref{eq:K-matrix}, \ie{}, the resulting set consists of all $F_k $ and $K_{m,n} $ such that $1 \leq k \leq N $ and $1 \leq n < m \leq N $. 
However, this is also not sufficient: even when applied to the elements of this extended set, the output of~\eqref{eq:prop:cont-formal-rels} contains the integrals of motion whose indices exceed~$N$. While these quantities are well-defined by~\cref{eq:Fk-def,eq:Kmn-cont} for any integer indices, we must express them purely in terms of the elements of the generating set. Fortunately, similarly to the Calogero--Moser case, this can be done using the Cayley--Hamilton theorem. This is the content of the following statement:          
\begin{prop}\label{prop:cont-closure-rels}
Consider the extended set of integrals of motion for the RS model, given by:
\begin{equation}\label{eq:gensetN-cont-extended}
	\genset^{(N)} \coloneqq \set{ F_n  }_{n=1}^N 
	\cup 
	\set{ K_{n,m} }_{1 \leq m < n \leq N },  
\end{equation}	
such that all functionally independent integrals~\eqref{eq:SN-fi} form a subset \emph{$ \genset^{(N)}_{\text{f.i.}} \subset \genset^{(N)}  $}. Then, the Poisson bracket between any pair of elements of $\genset^{(N)} $, given by the formulas~\eqref{eq:prop:cont-formal-rels}, can be expressed again purely in terms of elements of $\genset^{(N)} $, using the antisymmetry condition $K_{m,n} = - K_{n,m} $, and the following recursive formulas for $s>0 $:%
\begin{subequations}\label{eq:prop:cont-closure-rels}
\begin{align}
	F_{N+s} &= - \sum_{k=1}^N \frac{1}{k!} 
	\Bell_k \bigl( -0! F_1, \, -1! F_2, \, -2! F_3, \ldots, -(k-1)! F_k \bigr) F_{N+s-k},  
	\label{eq:prop:cont-closure-rels-F}
	\\
	K_{N+s, \ell } & = - \sum_{k=1}^N \frac{1}{k!} 
	\Bell_k \bigl( -0! F_1, \, -1! F_2, \, -2! F_3, \ldots, -(k-1)! F_k \bigr) K_{N+s-k, \ell }. 
	\label{eq:prop:cont-closure-rels-K}
\end{align}	
\end{subequations} 
where $\Bell_k ( \hyphen )   $ is the $k$-\emph{th} complete exponential Bell polynomial.  
\end{prop}
\begin{proof}
Let us denote by $A $ a generic $N \times N $ matrix with complex entries, and $a_k \coloneqq \trace(A^k) $. Then, it follows from the Cayley--Hamilton theorem that any $a_{N+s}$ with $s>0 $ can be expressed in terms of $ \set{a_1, \ldots, a_{N+s-1} } $ using the formula:
\begin{equation}\label{eq:Bell-fla-generic}
	a_{N+s} = - \sum_{k=1}^N \frac{1}{k!} 
	\Bell_k \bigl( -0! a_1, \, -1! a_2, \, -2! a_3, \ldots, -(k-1)! a_k \bigr) a_{N+s-k},
\end{equation}
see, \eg{}, \cite[appendix 1]{DGL2026Calogero} for the derivation. Applying this formula recursively starting from $s=1 $, one expresses any $a_{N+s} $ in terms of $ \set{a_1, \ldots, a_N} $. Hence, taking $A= L $ (the Lax matrix of the RS model), one readily obtains~\cref{eq:prop:cont-closure-rels-F}.  

Using the Cayley--Hamilton theorem again, one arrives at a similar expression for $J_{N+s} $ in terms of $ \set{ F_1, \ldots, F_N, J_1, \ldots, J_N } $. Indeed: 
\begin{equation}\label{eq:proof:cont-closure-rels-JN+s}
	J_{N+s} = \trace( X L^{N+s} )
	= - \sum_{k=1}^N \frac{1}{k!} 
	\Bell_k \bigl( -0! F_1, \, -1! F_2, \, -2! F_3, \ldots, -(k-1)! F_k \bigr) J_{N+s-k}. 
\end{equation}
Therefore, from~\eqref{eq:Kmn-cont}, \eqref{eq:prop:cont-closure-rels-F}, and~\eqref{eq:proof:cont-closure-rels-JN+s}, we obtain~\eqref{eq:prop:cont-closure-rels-K} in terms of the elements of~\eqref{eq:gensetN-cont-extended}. 
Note that $K_{\ell, N+s} $ can always be obtained from $K_{N+s, \ell} $ using the antisymmetry condition.    
\end{proof}

Therefore, using the findings above, we conclude with an explicit description of the symmetry algebra in the continuous-time case: 
\begin{prop}\label{prop:Symalg-N-cont}
The symmetry algebra $\Alg^{(N)} $ of the $N$-body RS model is a polynomial Poisson algebra, generated by the elements of the set $\genset^{(N)} $, defined by~\cref{eq:gensetN-cont-extended}. Its Poisson structure relations are given in eqs.~\eqref{eq:prop:cont-formal-rels}, provided the additional closure relations in~\cref{prop:cont-closure-rels}. Moreover, the set $ \braket{F_i}_{i=1}^N $ is an ideal of $\Alg^{(N)} $. Finally, the degree of the polynomial algebra $\Alg^{(N)} $ is $3N $.  	
\end{prop}
\begin{proof}
All the assertions of the~\namecref{prop:Symalg-N-cont} follow from the statements above except for the total degree and the ideal. The latter is clear from the right-hand side of~\cref{eq:prop:cont-formal-rels-FK}, which can be expressed only in terms of $\set{F_1, \ldots, F_N} $ using formula~\eqref{eq:prop:cont-closure-rels-F}.  
 Therefore, we are only left to prove that~\mbox{$\deg \Alg^{(N)} =3N $}.

For the two-body case, we have $\deg \Alg^{(2)} =6$ (see~\cref{ex:N=2-cont}), hence the statement holds true for $N=2$. To proceed, let us assume $N>2 $. We have: 
\begin{equation}\label{proof:Symalg-N-cont:degAlg-1}
    \deg\Alg^{(N)}=
    \max\Set{\max_{\substack{i=1,\ldots,N\\1\leq n<m\leq N}}\deg_{\genset^{(N)}}
    \{F_i, K_{m,n}\},
    \max_{\substack{1\leq j<i\leq N\\1\leq n<m\leq N}}\deg_{\genset^{(N)}}
    \{ K_{i,j} ,  K_{m,n}\} }.
\end{equation}
Clearly, $ \deg_{\genset^{(N)}} (F_k) =1 $ for $k=1,\ldots,N $  
and $  \deg_{\genset^{(N)}} ( K_{i,j} ) =1  $ for $1 \leq i,j \leq N $. Then, applying the formulas~\eqref{eq:prop:cont-closure-rels} recursively, and taking into account that $ \deg_{ \vec y }  \Bell_k ( \vec y) = k  $, we obtain for~$s>0 $:%
\begin{subequations}\label{proof:Symalg-N-cont:degsFandK}
\begin{align}
	 \deg_{\genset^{(N)}} ( F_{N+s} ) &= N+s,
	 \\ 
	  \deg_{\genset^{(N)}} ( K_{N+s, \ell } ) & = N+s. 
\end{align}	
\end{subequations}
Hence, let us estimate the maximal degree of the relations in~\cref{proof:Symalg-N-cont:degAlg-1}: 
\begin{subequations}
\begin{align}
\max_{\substack{i=1,\ldots,N\\1\leq n<m\leq N}}\deg_{\genset^{(N)}}
    \{  F_i ,  K_{m,n} \}
    &=
    \deg_{\genset^{(N)}} \{F_{N} ,  K_{N,N-1}\}
    =
    \deg_{\genset^{(N)}} ( F_{N+1} F_{2N-1} )  = 3N, 
\\ 
 \max_{\substack{1\leq j<i\leq N\\1\leq n<m\leq N}}\deg_{\genset^{(N)}}\pb*{K_{i,j}}{K_{m,n}}
    &=
    \deg_{\genset^{(N)}}\pb{K_{N,N-2}}{K_{N,N-1}} = \deg_{\genset^{(N)}} ( F_{N+1} K_{2N-1, N-2}) = 3N. 
\end{align}	
\end{subequations}
Therefore, we have $ \deg_{\genset^{(N)}} \Alg^{(N)} = 3N $ for $N>2$ as well.  
\end{proof}

Let us now illustrate~\cref{prop:Symalg-N-cont} by applying it to compute the symmetry algebra of the three-body model. 
\begin{example}[$N=3$]\label{ex:N=3-cont}
For the symmetry algebra $\Alg^{(3)} $ of the three-body RS model, we may rewrite the complete generating set~\eqref{eq:gensetN-cont-extended} as follows:
\begin{equation}\label{ex:N=3-cont:genset}
	\genset^{(3)} = \set{ F_1, F_2, F_3, K_{2,1}, K_{3,1}, K_{3,2} },
\end{equation} 
of which the subset of functionally independent integrals of motion $\genset^{(3)}_{\text{f.i.}} = \set{ F_1, F_2, F_3, K_{2,1}, K_{3,1} } $ consists of $2N-1= 5 $ elements, ensuring the maximal superintegrability of the model. Explicitly:%
\begin{subequations}\label{eq:proof:Symalg-N-cont:expl-Fs}
\begin{align}
  F_1 & = b_1 + b_2 + b_3, & 
  \\ 
  F_2 & = b_1^2 + b_2^2 + b_3^2 + 2 \chi^2 \left(
  \frac{b_1 b_2}{ (q_1 - q_2)^2 + \chi^2 }
  + \frac{b_1 b_3}{ (q_1 - q_3)^2 + \chi^2 } 
  + \frac{b_2 b_3}{ (q_2 - q_3)^2 + \chi^2 } 
  \right), & 
  \\ 
  \begin{split}
  	 F_3 & = b_1^3  + b_2^3 + b_3^3 + 
  	 3 \chi^2 \left(
  	    \frac{b_1 b_2 (b_1 + b_2 )}{ (q_1 - q_2)^2 + \chi^2 }
  	  + \frac{b_1 b_3 (b_1 + b_3 )}{ (q_1 - q_3)^2 + \chi^2 }
  	  + \frac{b_2 b_3 (b_2 + b_3 )}{ (q_2 - q_3)^2 + \chi^2 } 
  	 \right)
  	 \\
  	 & \quad\quad\quad\quad     
+ 3\chi^4 b_1 b_2 b_3 \frac{(q_1 - q_2)^2 + (q_2 - q_3)^2 + (q_1 - q_3)^2 + 2\chi^2}{ \left((q_1 - q_2)^2 + \chi^2\right )
 \left((q_2 - q_3)^2 + \chi^2 \right)
 \left((q_1 - q_3)^2 + \chi^2 \right)}. 
  \end{split}
\end{align}	
\end{subequations}
and:
\begin{subequations}\label{eq:proof:Symalg-N-cont:expl-Ks}
\begin{align}
  K_{2,1} &= ( \vec q \vdot \vec A) F_2 - ( \vec q \vdot \vec b) F_3, 
  \\ 
  \begin{split}
  	 K_{3,1} &  = F_2 \left( (\vec q \vdot \vec A) F_1 - 
  \frac{1}{2} (\vec q \vdot  \vec b) (F_1^2 -F_2 ) 
  + \frac{1}{6} (q_1 + q_2 + q_3 ) ( F_1^3 - 3 F_1 F_2 + 2 F_3 )    \right)
  \\ 
  &\quad  - ( \vec q \vdot \vec b) \left( F_1 F_3 
  - \frac{1}{2} F_2 ( F_1^2 - F_2) 
  + \frac{1}{6} F_1 ( F_1^3 - 3 F_1 F_2 + 2 F_3  ) \right), 
  \end{split}
  \\
 \begin{split}
 	K_{3,2} & = F_3 \left( ( \vec q \vdot \vec A) F_1 
 	- \frac{1}{2} ( \vec q \vdot \vec b) ( F_1^2 -F_2 )
 	+ \frac{1}{6} ( q_1 + q_2 + q_3 ) ( F_1^3 -3 F_1 F_2 + 2 F_3 ) \right)
 	\\
 	&\quad - (\vec q \vdot \vec A)  \left(
 	F_1 F_3 - \frac{1}{2} F_2 ( F_1^2 - F_2 )
 	+ \frac{1}{6} F_1 ( F_1^3 - 3 F_1 F_2 + 2 F_3 ) 
 	\right),
 \end{split}	
\end{align}	
\end{subequations}
where we define the auxiliary vector $\vec A = (A_1, A_2, A_3) $ through its components:
\begin{equation}\label{eq:vec-A-definition}
	A_i \coloneqq b_i^2 + \chi^2 \sum_{  k \ne i }^N \frac{ b_i b_k }{ (q_i - q_k)^2 + \chi^2 }, 
	\qquad
	i=1,2,3. 
\end{equation}

Clearly, the abelian subalgebra (corresponding to Liouville integrability) is, from~\cref{eq:prop:cont-formal-rels-FF}: 
\begin{equation}\label{eq:proof:Symalg-N-cont:abelian}
	\{ F_1, F_2 \} =0, 
	\qquad
	\{ F_1, F_3 \} =0, 
	\qquad
	\{ F_2, F_3 \} =0.
\end{equation}	
Moreover, by construction, $ H_+ = F_1  $ is a central element of $\Alg^{(3)} $:
\begin{equation}\label{eq:proof:Symalg-N-cont:central-element}
	\{ F_1, \hyphen \} =0. 
\end{equation}
Then, using~\eqref{eq:prop:cont-formal-rels-FK} and~\eqref{eq:prop:cont-formal-rels-KK}, we formally obtain:  
\begin{subequations}\label{eq:ex:N=3-cont:formal} 
\begin{equation}\label{eq:ex:N=3-cont:formalFK}
\begin{array}{lll}
	\{ F_2, K_{2,1} \} = 2 (F_3^2  - F_2 F_4 ),
	& 
	\{ F_2, K_{3,1} \} = 2 ( F_3 F_4 - F_2 F_5 ),
	&
	\{ F_2, K_{3,2} \} = 2 ( F_4^2 - F_3 F_5 ),
	\\
	\{ F_3, K_{2,1} \} = 3 ( F_3 F_4 - F_2 F_5 ), 
	& 
	\{ F_3, K_{3,1} \} = 3 ( F_4^2  - F_2 F_6 ), 
	&
	\{ F_3 , K_{3,2} \} = 3 ( F_4 F_5 - F_3 F_6 ),
\end{array}
\end{equation}
and: 
\begin{equation}\label{eq:ex:N=3-cont:formalKK}
\begin{aligned}
\{ K_{2,1}, K_{3,1} \} & = 
F_2 (3 K_{1,5} + 2 K_{4,2} + 4 K_{5,1})  + 2F_3 (2 K_{1,4} +  K_{2,3}) + 3 F_4 K_{3,1},
\\ 
\{ K_{2,1}, K_{3,2} \} & = 
F_3 (3 K_{1,5} + 4 K_{2,4} + 2 K_{4,2})  
+ F_2 (4 K_{5,2} + 3 K_{3,4}) 
+  F_4 ( 2K_{2,3} + 3 K_{4,1} ), 
\\
\{ K_{3,1} , K_{3,2} \} & = F_2 (3 K_{3,5} + 4 K_{6,2}) 
+ 2F_3 ( 2 K_{1,6} +  K_{4,3}) + 4 F_4 (K_{2,4} + K_{5,1}). 
\end{aligned}	
\end{equation}
\end{subequations} 
However, note that the relations~\eqref{eq:ex:N=3-cont:formal} do not yet define a finitely-generated polynomial Poisson algebra because they contain ``extra'' elements, which are not included in the generating set $\genset^{(3)} $ \eqref{ex:N=3-cont:genset}, namely $\set{ F_4, F_5, F_6, K_{4,1},  K_{1,4}, K_{1,5}, K_{5,1}, K_{2,3}, K_{4,2}, K_{2,4}, K_{4,3}, K_{3,4} , K_{5,2}, K_{3,5}, K_{1,6}, K_{6,2}  } $. Nevertheless, we are able to express them purely in terms of the elements of $\genset^{(3)} $ using~\cref{prop:cont-closure-rels}. Applying the formulas~\eqref{eq:prop:cont-closure-rels} recursively, we obtain, after the necessary algebraic manipulations: 
\begin{subequations}\label{eq:ex:N=3-cont:extra-els-expressed}
\begin{align}
	F_4 & = \frac{1}{6} F_1^4 - F_1^2 F_2 + \frac{4}{3} F_1 F_3 + \frac{1}{2} F_2^2, 
	\label{eq:ex:N=3-cont:extra-els-expressed-F4}
	\\
	F_5 & = \frac{5}{6} F_1^2 F_3 - \frac{5}{6} F_1^3 F_2 + \frac{1}{6} F_1^5 + \frac{5}{6} F_2 F_3, 
	\\
	F_6 & = \frac{1}{3} F_1^3 F_3 - \frac{1}{4} F_1^4 F_2 + \frac{1}{12} F_1^6 + F_1 F_2 F_3 - \frac{3}{4} F_1^2 F_2^2 + \frac{1}{4} F_2^3 + \frac{1}{3} F_3^2, 
	\\
	K_{1,4}= -K_{4,1} &= \frac{1}{2} (F_1^2 - F_2) K_{2,1} - F_1 K_{3,1}, 
	\\ 
	K_{1,5} = - K_{5,1} & = \frac{1}{3} (F_1^3 - F_3) K_{2,1} - \frac{1}{2} (F_1^2 + F_2) K_{3,1}, 
	\\
	K_{2,3} &= - K_{3,2}, 
	\\
	K_{4,2} & = \frac{1}{6} (-F_1^3 + 3 F_1 F_2 - 2 F_3) K_{2,1} + F_1 K_{3,2}, 
	\\
	K_{2,4} & = - K_{4,2}, 
	\\
	K_{4,3} & = \frac{1}{2}(F_1^2 - F_2)K_{3,2} + \frac{1}{6}(-F_1^3 + 3F_1F_2 - 2F_3)K_{3,1}, 
	\\  
	K_{5,2} & = -\frac{1}{6} F_1^4 K_{2,1} + \frac{1}{6} F_1^2  (3 F_2 K_{2,1} + 3 K_{3,2}) - \frac{1}{3} F_1 F_3 K_{2,1} + \frac{1}{2} F_2 K_{3,2},
	\\
	K_{3,5}  = - K_{5,3} &= -\frac{1}{3} F_1^3 K_{3,2} + \frac{1}{6} F_1^4 K_{3,1} - \frac{1}{2} F_1^2 F_2 K_{3,1} + \frac{1}{3} F_1 F_3 K_{3,1} + \frac{1}{3} F_3 K_{3,2},
	\\
	\begin{split}
		K_{1,6}   = - K_{6,1} & =  \frac{1}{12} F_1^4 K_{2,1} - \frac{1}{6} F_1^3 K_{3,1} + \frac{1}{2} F_1^2 F_2 K_{2,1}
		\\ &
		  \quad\quad\quad\quad - \frac{1}{6}F_1 ( 2 F_3 K_{2,1} +3 F_2 K_{3,1} )  
		   -    \frac{1}{4} F_2^2 K_{2,1} - \frac{1}{3} F_3 K_{3,1}, 
	\end{split}
	\\
	\begin{split}
		K_{6,2} & = -\frac{1}{12} F_1^5 K_{2,1} + \frac{1}{6} F_1^3  (F_2 K_{2,1} + K_{3,2}) 
		\\ & \quad\quad  - \frac{1}{6} F_1^2 F_3 K_{2,1} 
		+ \frac{1}{4}F_1 (F_2^2 K_{2,1} + 2 F_2 K_{3,2})  - \frac{1}{6} F_3 (F_2 K_{2,1} - 2 K_{3,2}). 
	\end{split}
\end{align}	
\end{subequations}  
Finally, substituting the closure relations~\eqref{eq:ex:N=3-cont:extra-els-expressed} into~\eqref{eq:ex:N=3-cont:formal}, we arrive at the complete set of structure relations among the generators of the symmetry algebra of the three-body RS model:
\begin{subequations}\label{eq:ex:N=3-cont:all-rels}
\begin{align}
\{ F_2, K_{2,1} \} & = - \frac{1}{3} (F_1^4 + 8 F_1 F_3) F_2 + 2 F_1^2 F_2^2 - F_2^3 + 2 F_3^2,
\\ 
\{ F_2, K_{3,1} \} & =  \frac{1}{3} \left( -F_1^5 F_2 + F_1^4 F_3 + 5 F_1^3 F_2^2 - 11 F_1^2 F_2 F_3 + 8 F_1 F_3^2 - 2 F_2^2 F_3 \right), 
\\
\begin{split}
	\{ F_2, K_{3,2} \} &=  \frac{1}{18} \bigl( F_1^8 - 12 F_1^6 F_2 + 10 F_1^5 F_3 
	 + 42 F_1^4 F_2^2 - 66 F_1^3 F_2 F_3 + 34 F_1^2 F_3^2 
	 \\ &\quad  - 36 F_1^2 F_2^3 + 48 F_1 F_2^2 F_3 + 9 F_2^4 - 30 F_2 F_3^2 \bigr), 
\end{split} 
\\
\{ F_3, K_{2,1} \} & =  \frac{1}{2} \left( -F_1^5 F_2 + F_1^4 F_3 + 5 F_1^3 F_2^2 - 11 F_1^2 F_2 F_3 + 8 F_1 F_3^2 - 2 F_2^2 F_3 \right), 
\\ 
\begin{split}
\{ F_3, K_{3,1} \} & = \frac{1}{12} \bigl( F_1^8 - 15 F_1^6 F_2 + 16 F_1^5 F_3 + 51 F_1^4 F_2^2 - 108 F_1^3 F_2 F_3 
 \\ & \quad  + 64 F_1^2 F_3^2 - 9 F_1^2 F_2^3 + 12 F_1 F_2^2 F_3 - 12 F_2 F_3^2 \bigr), 
\end{split}
\\
\begin{split}
\{ F_3, K_{3,2} \} & =  \frac{1}{12} \bigl( F_1^9 - 11 F_1^7 F_2 + 10 F_1^6 F_3 + 33 F_1^5 F_2^2 - 56 F_1^4 F_2 F_3 - 15 F_1^3 F_2^3 + 28 F_1^3 F_3^2  \\ & \quad  + 12 F_1^2 F_2^2 F_3 + 4 F_1 F_2 F_3^2 + 6 F_2^3 F_3 - 12 F_3^3 \bigr), \label{eq:ex:N=3-cont:F3-K32} 	
\end{split}
\\
\begin{split}
 \{ K_{2,1}, K_{3,1} \} & =  \frac{1}{6} \bigl( 3 F_1^4 K_{3,1} - 4 F_1^3 F_2 K_{2,1} - 15 F_1^2 F_2 K_{3,1} + 12 F_1^2 F_3 K_{2,1} + 6 F_1 F_2^2 K_{2,1} \\ & \quad  + 12 F_1 F_2 K_{3,2} + 12 F_2^2 K_{3,1} - 14 F_2 F_3 K_{2,1} - 12 F_3 K_{3,2} \bigr), 
\end{split}
\\ 
\begin{split}
\{ K_{2,1}, K_{3,2} \} & =  \frac{1}{12} \bigl( -3 F_1^6 K_{2,1} + 6 F_1^5 K_{3,1} + 13 F_1^4 F_2 K_{2,1} - 4 F_1^4 K_{3,2} - 30 F_1^3 F_2 K_{3,1} - 8 F_1^3 F_3 K_{2,1} 
\\ & \quad  - 3 F_1^2 F_2^2 K_{2,1} + 30 F_1^2 F_2 K_{3,2} + 30 F_1^2 F_3 K_{3,1} - 4 F_1 F_2 F_3 K_{2,1} - 56 F_1 F_3 K_{3,2} 
\\ & \quad  + 9 F_2^3 K_{2,1} + 30 F_2^2 K_{3,2} - 6 F_2 F_3 K_{3,1} - 4 F_3^2 K_{2,1} \bigr),	
\end{split}
\\
\begin{split}
\{ K_{3,1}, K_{3,2} \} & =  \frac{1}{18} \bigl( -2 F_1^7 K_{2,1} + 6 F_1^6 K_{3,1} - 12 F_1^5 K_{3,2} - 21 F_1^4 F_2 K_{3,1} - 2 F_1^4 F_3 K_{2,1} + 42 F_1^3 F_2^2 K_{2,1} 
\\ &\quad  + 66 F_1^3 F_2 K_{3,2} + 30 F_1^3 F_3 K_{3,1} - 45 F_1^2 F_2^2 K_{3,1} - 72 F_1^2 F_2 F_3 K_{2,1} - 78 F_1^2 F_3 K_{3,2} 
\\ &\quad + 48 F_1 F_2 F_3 K_{3,1} + 40 F_1 F_3^2 K_{2,1} + 18 F_2^3 K_{3,1} - 6 F_2^2 F_3 K_{2,1} + 24 F_2 F_3 K_{3,2} - 36 F_3^2 K_{3,1} \bigr).
\end{split}
\end{align}	
\end{subequations}   
Note that, as expected, we have $ \deg_{ \genset^{(3)} } \Alg^{(3)} = 9 $, see, \eg{},~\cref{eq:ex:N=3-cont:F3-K32}. Moreover, the difference of the cardinalities of the generating set and the subset of functionally independent integrals is given by:
\begin{equation}
	\bigl| \genset^{(3)} \bigr| - 
	\bigl| \genset^{(3)}_{\text{f.i.}} \bigr| = 6 - 5 =1, 
\end{equation} 
which means that in this case we should have one functional relation among the elements of~$\genset^{(3)}$. Note that this is in line with the counting~\eqref{eq:funrel-counting}. Indeed, from~\cref{eq:funrel-cont-G}, the functional relation, which connects $K_{3,2} $ to the elements of $\genset^{(3)}_{\text{f.i.}} $, is given by:
\begin{equation}\label{eq:ex:N=3-cont:funrel} 
	F_{3} K_{3,1} - F_{2} K_{3,2} - \left( \frac{1}{6} F_1^4 - F_1^2 F_2 + \frac{4}{3} F_1 F_3 + \frac{1}{2} F_2^2  \right) K_{2,1} =0, 
\end{equation}
where we used the additional closure relation~\eqref{eq:ex:N=3-cont:extra-els-expressed-F4}.  
\end{example}

\subsection{Deformed symmetry algebra in discrete-time case}
\label{sec:Symalg-discr}
We now proceed to the discrete-time case. Our goal is to obtain the complete structure of the symmetry algebra of the discrete RS model (which can be done in a similar fashion to the continuous-time case), and to show that this symmetry algebra is a deformation of its continuous-time counterpart. Since the latter phenomenon emerges already in the two-body case, let us again start by presenting an example. 

\begin{example}[$N=2$]\label{ex:N=2-discr}
	For the discrete-time two-body model, the symmetry algebra is generated by the set (see~\cref{prop:discrete-system-is-MS}): 
\begin{equation}
	\widetilde{\genset}^{(2)} = \widetilde{\genset}^{(2)}_{\text{f.i.}}
	= \set{  F_1, F_2, \K_{2,1} },
\end{equation}	 
where $F_1 $ and $F_2 $ remain the same as in the continuous case, see~\cref{eq:ex:N=2-cont-Hplus-explicit-expression,eq:ex:N=2-cont-explicit-expressions}, and the modified integral of motion is given by~\eqref{eq:discrete-additional-constants-def}; it can be rewritten explicitly in the two-body case as follows:
\begin{equation}
	\K_{2,1} = K_{2,1} - \frac{1}{4} \imath \chi^{-1} h  \, (q_1 - q_2 ) ( b_1 - b_2 ) ( F_1^2 - F_2)^2, 
\end{equation} 
where $K_{2,1} $ is a constant of motion for the continuous-time system, see~\cref{eq:ex:N=2-cont-explicit-expressions}. Computing the Poisson brackets, we arrive at the following structure relations: 
\begin{subequations}\label{eq:ex:N=2-discr-structure-relations}
\begin{align}
	\{ F_1, F_2 \} & =0, \label{eq:ex:N=2-discr-structure-relations-FF}
	\\ 
	\{ F_1, \K_{2,1} \} & = \frac{1}{4} \imath h \chi^{-1}   \, 
	(2F_2 -F_1^2 ) ( F_1^2 - F_2)^2, 
	\label{eq:ex:N=2-discr-structure-relations-F1K}
	\\ 
	\{ F_2, \K_{2,1} \} & = \frac{1}{2}  ( F_1^2 - 2 F_2 ) ( F_1^2 - F_2 )^2 
	- \frac{1}{2} \imath h \chi^{-1}  \, F_1 ( F_1^2 - 2 F_2 )
	(F_1^2 - F_2 )^2.  
	\label{eq:ex:N=2-discr-structure-relations-F2K}
\end{align}	
\end{subequations}
Comparing~\eqref{eq:ex:N=2-discr-structure-relations} with the analogous continuous-time relations~\eqref{eq:ex:N=2-cont-structure-relations} from~\cref{ex:N=2-cont}, we observe that:
\begin{itemize}
	\item The relation~\eqref{eq:ex:N=2-discr-structure-relations-FF} remains unaltered under discretization. 
	\item The relation~\eqref{eq:ex:N=2-discr-structure-relations-F1K} is now nonzero, which is aligned with the fact that in the discrete-time case there is no Hamiltonian in the usual sense, \ie{}, $F_1 $ is no longer a central element of the symmetry algebra (see the discussion in the introductory~\cref{sec:SI-review}). However, in the continuous limit we recover the conservation of $K_{2,1} $ for the continuous-time RS~model:
	\begin{equation}
		\lim_{h \to 0 } \{ F_1, \K_{2,1} \} =0. 
	\end{equation}
	\item The relation~\eqref{eq:ex:N=2-discr-structure-relations-F2K} is deformed compared to its continuous counterpart, and the time step $h $ plays the r\^ole of the deformation parameter. Indeed, evaluating~\eqref{eq:ex:N=2-discr-structure-relations-F2K} at $h=0 $, we see that the bracket reduces to its continuous counterpart~\eqref{eq:ex:N=2-cont-structure-relations}: 
	\begin{equation}
		\lim_{h \to 0} \{ F_2, \K_{2,1} \} 
		= \frac{1}{2} ( F_1^2 - F_2 )^2 ( F_1^2 - 2 F_2 )  = 
		\{ F_2, K_{2,1} \}. 
	\end{equation}
\end{itemize} 
Thus, we observe that in the two-body case, discretization deforms a symmetry algebra of the model. The continuous limit is given simply by evaluating the expression at $h=0 $, without the necessity of any parameter rescaling. 
 Moreover, discretization increases the degree of a polynomial algebra: while the degree in the continuous case was $6$, here the degree is $7 $ (\cf{}~\cref{ex:N=2-cont}), provided by the additional multiplier $F_1 $ in the deformation term in~\cref{eq:ex:N=2-discr-structure-relations-F2K}. Below we prove that this phenomenon remains true for the arbitrary-$N$ case. 
 
 A similar pattern was observed in the (nonrelativistic) Calogero--Moser model; see \cite{DGL2026Calogero}. We will discuss their common features and their differences later in~\cref{sec:concl}. 
\end{example}

Let us now proceed to the arbitrary $N$-body case. We start by extending~\cref{prop:cont-formal-rels} to the discrete-time model. 
 
\begin{prop}\label{prop:discr-formal-rels}
The integrals of motion of the discrete-time RS model obey the following relations:%
\begin{subequations}\label{eq:prop:discr-formal-rels}%
\begin{align}%
 \{ F_k, F_m \} & = 0, 
\label{eq:prop:discr-formal-rels-FF}
 \\ 
	\{ F_k, \K_{m,n} \} &= k \left( F_{m+1} F_{k+n} - F_{n+1} F_{k+m} \right)
      - \imath k h \chi^{-1} \left( F_{m+1} F_{k+n+1} - F_{n+1} F_{k+m+1} \right), 
	\label{eq:prop:discr-formal-rels-FK}
	\\
	\begin{split}
		\{ \K_{i,j}, \K_{m,n} \} &= 
	 F_{i+1} \bigl( (m+1) \K_{n, j+m} + (n+1) \K_{j+n, m}  \bigl) 
	 + F_{j+1} \bigl( (m+1) \K_{i+m,n} +(n+1) \K_{m, i+n}  \bigr) \\ 
	 & \quad + F_{m+1} \bigl( (i+1) \K_{i+n, j} + (j+1) \K_{i, j+n}  \bigr)
	 + F_{n+1} \bigl( (i+1) \K_{j, i+m} + (j+1) \K_{j+m,i} \bigr)
	 \\
	 & \quad - \imath  h \chi^{-1}  \bigl[
	 F_{i+1} \bigl( (m+1) \K_{n, j+m+1} + (n+1) \K_{j+n+1, m}  \bigl) 
	 + F_{j+1} \bigl( (m+1) \K_{i+m+1,n} +(n+1) \K_{m, i+n+1}  \bigr) 
	 \\
	 & \quad\quad  
	 + F_{m+1} \bigl( (i+1) \K_{i+n+1, j} + (j+1) \K_{i, j+n+1}  \bigr)
	 + F_{n+1} \bigl( (i+1) \K_{j, i+m+1} + (j+1) \K_{j+m+1,i} \bigr) 
	 \bigr]. 
	 \label{eq:prop:discr-formal-rels-KK}
	\end{split}
\end{align}	
\end{subequations}
\end{prop}
We remark that in the continuous limit, evaluating~eqs.~\eqref{eq:prop:discr-formal-rels} at $h=0 $ and taking into account that $\lim_{h \to 0 } \K_{i,j} = K_{i,j}  $ by~\cref{prop:discrete-additional-constants}, one recovers the continuous-time relations~\eqref{eq:prop:cont-formal-rels}. However, since the same reasoning as in the continuous case is applicable also here, we have to perform additional constructions to show that one can derive a finitely-generated Poisson algebra using the relations~\eqref{eq:prop:discr-formal-rels}. For now, we proceed to the proof of the statement.   

\begin{proof}[\cref{prop:discr-formal-rels}]
Using~\cref{eq:Kdiscr-via-flows} and the properties of the Poisson bracket, we can expand the left-hand side of eqs.~\eqref{eq:prop:discr-formal-rels} as follows:
\begin{subequations}\label{eq:proof:discr-formal-rels:expanded}
\begin{align}
	\{ F_k, \K_{m,n} \}  = \{ F_k, K_{m,n} \}  - & \imath h \chi^{-1}  \{ F_k,  K_{m+1, n+1}^{(0)} \},  
	\\ 
	\begin{split}
		\{ \K_{i,j}, \K_{m,n} \} = \{ K_{i,j}, K_{m,n} \} 
	- & \imath h \chi^{-1} \{ K_{i,j}, K_{m+1, n+1}^{(0)} \} 
	- \imath h \chi^{-1} \{ K_{i+1, j+1}^{(0)}, K_{m,n} \} 
	\\ 
	- & h^2 \chi^{-2} \{ K_{i+1, j+1}^{(0)}, K_{m+1, n+1}^{(0)} \}.  
	\end{split}
\end{align}	
\end{subequations}
Hence, we need to compute the auxiliary relations of the form $  \{ F_k,  K_{m, n}^{(0)} \} $, $  \{ K_{i,j}, K_{m, n}^{(0)} \} $ and $  \{ K_{i, j}^{(0)}, K_{m, n}^{(0)} \} $. In this regard, using~\cref{eq:JFandJJ,eq:Ks-new-def}, and recombining the elements, one can prove that, in general:
\begin{subequations}\label{eq:proof:discr-formal-rels:higher-flows}
\begin{align}
\{ F_k, K_{m, n}^{(a)} \} & = k ( F_{a+m} F_{k+n} - F_{a+n} F_{k+m} ),
\\
\begin{split}
\{ K_{i, j}^{(a)}, K_{m,n}^{(b)}  \} & =  
      F_{a+i}  \bigl( (m+1) K_{n, j+m}^{(b)} + (n+1) K_{j+n, m}^{(b)} \bigr)
    + F_{a+j}  \bigl( (m+1) K_{i+m,n}^{(b)} +(n+1) K_{m, i+n}^{(b)}  \bigr)  
    \\
    &\quad + F_{b+m}  \bigl( (i+1) K_{i+n, j}^{(a)} + (j+1) K_{i, j+n}^{(a)}  \bigr) 
    + F_{b+n} \bigl( (i+1) K_{j, i+m}^{(a)} + (j+1) K_{j+m,i}^{(a)} \bigr)
    \\
    &\quad + (a-1) \bigl( F_{b+m} K^{(a+n)}_{i, j} + F_{b+n} K^{(a+m)}_{j,i} \bigr)
    + (b-1) \bigl( F_{a+i} K^{(b+j)}_{n,m} + F_{a+j} K^{(b+i)}_{m,n} \bigr). 
\end{split}
\end{align}	
\end{subequations} 
These relations constitute a generalization of the formulas~\eqref{eq:prop:cont-formal-rels}, which are recovered from~\eqref{eq:proof:discr-formal-rels:higher-flows} by setting $a = b = 1 $. Using eqs.~\eqref{eq:proof:discr-formal-rels:higher-flows} to compute the brackets in~\eqref{eq:proof:discr-formal-rels:expanded}, performing the necessary algebraic manipulations, and recombining the elements using~\eqref{eq:Kdiscr-via-flows} again, we obtain the right-hand sides of~eqs.~\eqref{eq:prop:discr-formal-rels} and prove the statement.   
\end{proof}

Now we can generalize~\cref{prop:cont-closure-rels} to the discrete case:

\begin{prop}\label{prop:discr-closure-rels}
Consider the extended set of integrals of motion for the discrete RS model:
\begin{equation}\label{eq:gensetN-discr-extended}
	\widetilde{\genset}^{(N)} \coloneqq \set{ F_n  }_{n=1}^N 
	\cup 
	\set{ \K_{n,m} }_{1 \leq m < n \leq N },  
\end{equation}	
such that all functionally independent integrals~\eqref{eq:discr-SN-fi} form a subset \emph{$ \widetilde{\genset}^{(N)}_{\text{f.i.}} \subset \widetilde{\genset}^{(N)}  $}. Then, the Poisson bracket between any pair of elements of $\widetilde{\genset}^{(N)} $, given by the formulas~\eqref{eq:prop:discr-formal-rels}, can be expressed again purely in terms of elements of $\widetilde{\genset}^{(N)} $, using the antisymmetry condition $\K_{m,n} = - \K_{n,m} $, and the following recursive formulas for $s>0 $:%
\begin{subequations}\label{eq:prop:discr-closure-rels}
\begin{align}
	F_{N+s} &= - \sum_{k=1}^N \frac{1}{k!} 
	\Bell_k \bigl( -0! F_1, \, -1! F_2, \, -2! F_3, \ldots, -(k-1)! F_k \bigr) F_{N+s-k},  
	\label{eq:prop:discr-closure-rels-F}
	\\
	\K_{N+s, \ell } & = - \sum_{k=1}^N \frac{1}{k!} 
	\Bell_k \bigl( -0! F_1, \, -1! F_2, \, -2! F_3, \ldots, -(k-1)! F_k \bigr) \K_{N+s-k, \ell }. 
	\label{eq:prop:discr-closure-rels-K}
\end{align}	
\end{subequations} 
where $\Bell_k ( \hyphen )   $ is the $k$-\emph{th} complete exponential Bell polynomial.  
\end{prop}
\begin{proof}
The proof follows the same steps as in the continuous case. 	
\end{proof}

Ultimately, these constructions enable us to extend~\cref{prop:Symalg-N-cont} to the discrete case, thereby obtaining the complete description of the symmetry algebra of the discrete RS model.%
\begin{prop}\label{prop:Symalg-N-discr}
The symmetry algebra $\widetilde{\Alg}^{(N)} $ of the discrete $N$-body RS model is a polynomial Poisson algebra, generated by the elements of the set $\widetilde{\genset}^{(N)} $, defined by~\cref{eq:gensetN-discr-extended}. Its Poisson structure relations are given in eqs.~\eqref{eq:prop:discr-formal-rels}, provided the additional closure relations in~\cref{prop:discr-closure-rels}. Furthermore, the set $ \braket{F_i}_{i=1}^N $ is an ideal of $\widetilde{\Alg}^{(N)} $. Finally, the degree of the polynomial algebra $\widetilde{\Alg}^{(N)} $ is~\mbox{$3N+1$}.  
\end{prop}
\begin{proof}
As before, we are only left to prove the degree assertion. For $N=2 $, it clearly holds; see~\cref{ex:N=2-discr}; hence let us assume $N>2 $. 
Note that due to the terms proportional to $h $ in relations~\eqref{eq:prop:discr-formal-rels}, the maximal indices of $F_k $ and $\K_{m,n} $ on the right-hand side were increased by one compared to the continuous-time case, \cf{} eqs.~\eqref{eq:prop:cont-formal-rels}. Taking this into account, we can apply the same reasoning as in the proof of~\cref{prop:Symalg-N-cont}, \mm{}. In this way, we obtain that
$ \deg_{\widetilde \genset^{(N)}} \widetilde \Alg^{(N)} = 3N+1 $.   
\end{proof}

Finally, the following observation is the main qualitative result of our symmetry algebra study. 

\begin{corollary}\label{cor:deformation}
	The symmetry algebra $\widetilde{\Alg}^{(N)}$ for the discrete-time RS model is a deformation of its continuous counterpart $\Alg^{(N)} $, with $h$ playing the r\^ole of the deformation parameter; that is, the structure relations~\eqref{eq:prop:discr-formal-rels} reduce to~\eqref{eq:prop:cont-formal-rels} in the continuous limit $h \to 0 $. Moreover, discretization increases the degree of the symmetry algebra by one (from $3N $ in the continuous-time case to $3N+1 $ in the discrete-time case).    
\end{corollary}

We conclude this section by presenting an application of the obtained statements to the three-body discrete-time RS model. 
\begin{example}[$N=3$]
The extended generating set~\eqref{eq:gensetN-discr-extended} for the symmetry algebra $\widetilde{\Alg}^{(3)} $ is given by:
\begin{equation}\label{ex:N=3-discr:genset}
	\widetilde{\genset}^{(3)} = \set{ F_1, F_2, F_3, \K_{2,1}, \K_{3,1}, \K_{3,2} },
\end{equation} 
with $\widetilde{\genset}^{(3)}_{\text{f.i.}} = \set{ F_1, F_2, F_3, \K_{2,1}, \K_{3,1} } $ consisting of $2N-1= 5 $ elements, providing maximal superintegrability of the discrete dynamical system. The expressions~\eqref{eq:proof:Symalg-N-cont:expl-Fs} for the Poisson-commuting invariants remain unchanged under discretization, while the modified additional constants can be expanded as follows:
\begin{subequations}\label{eq:proof:Symalg-N-discr:expl-Ks} 
\begin{align}
\begin{split}
\K_{2,1} & = K_{2,1} - \imath h \chi^{-1} 
	\left( ( \vec q \vdot \vec A) ( F_1 F_2 - F_3 )
	+\frac{1}{2} ( \vec q \vdot \vec b ) F_2 ( F_2 - F_1^2 ) \right. 
	\\ & \left. \qquad \quad
	+ \frac{1}{6} (q_1 + q_2 + q_3 ) F_2 ( F_1^3  - 3F_1 F_2 + 2 F_3)
	\right),  	
\end{split} 
\\
\begin{split}
\K_{3,1} & = K_{3,1} - \imath h \chi^{-1}  \left( 
 ( \vec q \vdot \vec A) \left( \frac{3}{2} F_1^2 F_2 - \frac{1}{6} F_1^4 - \frac{4}{3} F_1 F_3  \right) 
 + \frac{1}{3} ( \vec q \vdot \vec b ) F_2 ( F_3 - F_1^3 ) \right. 
 \\
 &\qquad \quad   
 \left. 
 + \frac{1}{6} ( q_1 + q_2 + q_3)  	
\, F_1 F_2 \left( F_1^3 - 3 F_1 F_2 + 2 F_3 \right) 
\right),  
\end{split}
\\
\begin{split}
\K_{3,2} & = K_{3,2} - \imath h \chi^{-1} \left(   ( \vec q \vdot \vec A) \left( - \frac{1}{6} F_1^5 + F_1^3 F_2 
- \frac{1}{2} F_1 F_2^2 - \frac{5}{6} F_1^2 F_3 
+ \frac{1}{2} F_2 F_3 \right)
\right. 
\\
& \qquad \quad + ( \vec q \vdot \vec b) 
\left( \frac{1}{12} F_1^6 - \frac{7}{12} F_1^4 F_2 
+ \frac{3}{4} F_1^2 F_2^2 - \frac{1}{4} F_2^3 
+\frac{1}{3} F_1^3 F_3 - \frac{2}{3} F_1 F_2 F_3 + \frac{1}{3} F_3^2  \right)
\\
& \left. \qquad \quad + \frac{1}{6} (q_1 + q_2 + q_3) 
\left( F_1^2 F_2 - \frac{1}{6} F_1^4 - \frac{1}{2}F_2^2 - \frac{1}{3} F_1 F_3 \right) 
( F_1^3 - 3 F_1 F_2 + 2 F_3 )
\right),  	
\end{split}
\end{align}	
\end{subequations}
where $K_{2,1}, K_{3,1} $, and $K_{3,2}$ are given by~\eqref{eq:proof:Symalg-N-cont:expl-Ks}, and $\vec A $ is defined by~\cref{eq:vec-A-definition}.  

Let us proceed to the structure relations. First, as before, we have:
\begin{equation}
	\{ F_1, F_2 \} =0, 
	\qquad
	\{ F_1, F_3 \} =0, 
	\qquad
	\{ F_2, F_3 \} =0. 
\end{equation}
Then, as in~\cref{ex:N=3-cont}, we first obtain formal relations using~\cref{prop:discr-formal-rels}, and then express all elements via~\cref{prop:discr-closure-rels}. After the algebraic manipulations, we arrive at the following relations, completely defining the symmetry algebra of the discrete three-body RS~model:   
\begin{subequations} \label{eq:ex:N=3-discr-allrels}
\begin{align}
\{ F_1, \K_{2,1} \} & = \frac{1}{6} \imath h \chi^{-1} ( 3 F_2^3 - 6 F_1^2 F_2^2 + F_1 F_2 (F_1^3 + 8 F_3)  - 6 F_3^2 ), \label{eq:ex:N=3-discr-allrels-F1-K21}
\\
\{ F_1, \K_{3,1} \} & = \frac{1}{6} \imath h \chi^{-1} \left( F_1^5 F_2 - F_1^4 F_3 - 5 F_1^3 F_2^2 + 11 F_1^2 F_2 F_3 - 8 F_1 F_3^2 + 2 F_2^2 F_3 \right),\label{eq:ex:N=3-discr-allrels-F1-K31}
\\
\begin{split}
\{ F_1, \K_{3,2} \} & = \frac{1}{36} \imath h \chi^{-1} \bigl( -F_1^8 + 12 F_1^6 F_2 - 10 F_1^5 F_3 - 42 F_1^4 F_2^2 + 66 F_1^3 F_2 F_3 + 36 F_1^2 F_2^3  
\\ &\quad 
- 34 F_1^2 F_3^2 - 48 F_1 F_2^2 F_3 - 9 F_2^4 + 30 F_2 F_3^2 \bigr), 
\label{eq:ex:N=3-discr-allrels-F1-K32}
\end{split}
\\
\begin{split}
\{ F_2, \K_{2,1} \} & = - \frac{1}{3} (F_1^4 + 8 F_1 F_3) F_2 + 2 F_1^2 F_2^2 - F_2^3 + 2 F_3^2 + \frac{1}{3} \imath h \chi^{-1} \bigl( F_1^5 F_2 - F_1^4 F_3 
\\ &\quad  - 5 F_1^3 F_2^2 + 11 F_1^2 F_2 F_3 - 8 F_1 F_3^2 + 2 F_2^2 F_3 \bigr),
\end{split} 
\\ 
\begin{split}
\{ F_2, \K_{3,1} \} & =  \frac{1}{3} \left( -F_1^5 F_2 + F_1^4 F_3 + 5 F_1^3 F_2^2 - 11 F_1^2 F_2 F_3 + 8 F_1 F_3^2 - 2 F_2^2 F_3 \right)
\\ & \quad 
+  \frac{1}{18} \imath h \chi^{-1} \bigl( -F_1^8 + 15 F_1^6 F_2 - 16 F_1^5 F_3 - 51 F_1^4 F_2^2 + 108 F_1^3 F_2 F_3 + 9 F_1^2 F_2^3 
\\ & \quad 
- 64 F_1^2 F_3^2 - 12 F_1 F_2^2 F_3 + 12 F_2 F_3^2 \bigr),   
\end{split}
\\
\begin{split}
	\{ F_2, \K_{3,2} \} &=  \frac{1}{18} \bigl( F_1^8 - 12 F_1^6 F_2 + 10 F_1^5 F_3 
	 + 42 F_1^4 F_2^2 - 66 F_1^3 F_2 F_3 + 34 F_1^2 F_3^2 
	 \\ &\quad  - 36 F_1^2 F_2^3 + 48 F_1 F_2^2 F_3 + 9 F_2^4 - 30 F_2 F_3^2 \bigr) 
	 + \frac{1}{18} \imath h \chi^{-1} \bigl( -F_1^9 + 11 F_1^7 F_2 
	 \\ &\quad  - 10 F_1^6 F_3 - 33 F_1^5 F_2^2 + 56 F_1^4 F_2 F_3 + 15 F_1^3 F_2^3 - 28 F_1^3 F_3^2 - 12 F_1^2 F_2^2 F_3 
	 \\ &\quad  - 4 F_1 F_2 F_3^2 - 6 F_2^3 F_3 + 12 F_3^3 \bigr),  
	 \label{eq:ex:N=3-discr-allrels-F2-K32}
\end{split} 
\\
\begin{split}
\{ F_3, \K_{2,1} \} & =  \frac{1}{2} \left( -F_1^5 F_2 + F_1^4 F_3 + 5 F_1^3 F_2^2 - 11 F_1^2 F_2 F_3 + 8 F_1 F_3^2 - 2 F_2^2 F_3 \right)
\\ &\quad 
+ \frac{1}{4} \imath h \chi^{-1} \bigl( F_1^6 F_2 - 2 F_1^5 F_3 - 3 F_1^4 F_2^2 + 14 F_1^3 F_2 F_3 - 9 F_1^2 F_2^3 - 10 F_1^2 F_3^2 
 \\ &\quad  + 12 F_1 F_2^2 F_3 + 3 F_2^4 - 6 F_2 F_3^2 \bigr),	
\end{split}
\\ 
\begin{split}
\{ F_3, \K_{3,1} \} & = \frac{1}{12} \bigl( F_1^8 - 15 F_1^6 F_2 + 16 F_1^5 F_3 + 51 F_1^4 F_2^2 - 108 F_1^3 F_2 F_3 
 + 64 F_1^2 F_3^2 - 9 F_1^2 F_2^3  \\ & \quad  + 12 F_1 F_2^2 F_3 - 12 F_2 F_3^2 \bigr)
 + \frac{1}{12} \imath h \chi^{-1} \bigl( -F_1^9 + 12 F_1^7 F_2 
  - 13 F_1^6 F_3 - 33 F_1^5 F_2^2 \\ &\quad  + 72 F_1^4 F_2 F_3 - 6 F_1^3 F_2^3 - 40 F_1^3 F_3^2 + 15 F_1^2 F_2^2 F_3 - 12 F_1 F_2 F_3^2 + 6 F_2^3 F_3 \bigr), 
\end{split}
\\
\begin{split}
\{ F_3, \K_{3,2} \} & =  \frac{1}{12} \bigl( F_1^9 - 11 F_1^7 F_2 + 10 F_1^6 F_3 + 33 F_1^5 F_2^2 - 56 F_1^4 F_2 F_3 - 15 F_1^3 F_2^3 + 28 F_1^3 F_3^2  \\ & \quad  + 12 F_1^2 F_2^2 F_3 + 4 F_1 F_2 F_3^2 + 6 F_2^3 F_3 - 12 F_3^3 \bigr)
 + \frac{1}{24} \imath h \chi^{-1} \bigl( -F_1^{10} + 9 F_1^8 F_2 
 \\ & \quad - 10 F_1^7 F_3 - 12 F_1^6 F_2^2 + 36 F_1^5 F_2 F_3 - 48 F_1^4 F_2^3 - 22 F_1^4 F_3^2 + 90 F_1^3 F_2^2 F_3 
 \\ & \quad + 45 F_1^2 F_2^4 - 72 F_1^2 F_2 F_3^2 - 60 F_1 F_2^3 F_3 - 9 F_2^5 + 24 F_1 F_3^3 + 30 F_2^2 F_3^2 \bigr), 
 \label{eq:ex:N=3-discr-allrels-F3-K32}
\end{split}
\\
\begin{split}
 \{ \K_{2,1}, \K_{3,1} \} & =  \frac{1}{6} \bigl( 3 F_1^4 \K_{3,1} - 4 F_1^3 F_2 \K_{2,1} - 15 F_1^2 F_2 \K_{3,1} + 12 F_1^2 F_3 \K_{2,1} + 6 F_1 F_2^2 \K_{2,1} \\ & \quad  + 12 F_1 F_2 \K_{3,2} + 12 F_2^2 \K_{3,1} - 14 F_2 F_3 \K_{2,1} - 12 F_3 \K_{3,2} \bigr) 
 \\ & \quad 
 + \frac{1}{4} \imath h \chi^{-1} \biggl( \biggl( F_1^6 - \frac{16}{3} F_1^4 F_2 + \frac{8}{3} F_1^3 F_3 + 7 F_1^2 F_2^2 - \frac{20}{3} F_1 F_2 F_3 - 4 F_2^3 + \frac{16}{3} F_3^2 \biggr) \K_{2,1} 
 \\ & \quad  + ( -2 F_1^5 + 10 F_1^3 F_2 - 8 F_1^2 F_3 - 4 F_1 F_2^2 + 4 F_2 F_3 ) \K_{3,1}
  \\ & \quad  + \frac{4}{3} \bigl( F_1^4 - 6 F_1^2 F_2 + 8 F_1 F_3 - 3 F_2^2 \bigr) \K_{3,2} \biggr), 
\end{split}
\\ 
\begin{split}
\{ \K_{2,1}, \K_{3,2} \} & =  \frac{1}{12} \bigl( -3 F_1^6 \K_{2,1} + 6 F_1^5 \K_{3,1} + 13 F_1^4 F_2 \K_{2,1} - 4 F_1^4 \K_{3,2} - 30 F_1^3 F_2 \K_{3,1} - 8 F_1^3 F_3 \K_{2,1} 
\\ & \quad  - 3 F_1^2 F_2^2 \K_{2,1} + 30 F_1^2 F_2 \K_{3,2} + 30 F_1^2 F_3 \K_{3,1} - 4 F_1 F_2 F_3 \K_{2,1} - 56 F_1 F_3 \K_{3,2} 
\\ & \quad  + 9 F_2^3 \K_{2,1} + 30 F_2^2 \K_{3,2} - 6 F_2 F_3 \K_{3,1} - 4 F_3^2 \K_{2,1} \bigr)
+ \frac{1}{9} \imath h \chi^{-1} \biggl( \biggl( F_1^7 - \frac{3}{2} F_1^5 F_2 
\\ & \quad + \frac{1}{4} F_1^4 F_3 - 12 F_1^3 F_2^2 + \frac{57}{2} F_1^2 F_2 F_3 - \frac{9}{2} F_1 F_2^3 - 17 F_1 F_3^2 + \frac{21}{4} F_2^2 F_3 \biggr) \K_{2,1} 
\\ & \quad + \left( -\frac{9}{4} F_1^6 + \frac{27}{4} F_1^4 F_2 - 9 F_1^3 F_3 + \frac{81}{4} F_1^2 F_2^2 - 27 F_1 F_2 F_3 - \frac{27}{4} F_2^3 + 18 F_3^2 \right) \K_{3,1} 
\\ & \quad + 3 \biggl( F_1^5 - 5 F_1^3 F_2 + \frac{13}{2} F_1^2 F_3 - 3 F_1 F_2^2 + \frac{1}{2} F_2 F_3 \biggr) \K_{3,2} \biggr),	
\end{split}
\\
\begin{split}
\{ \K_{3,1}, \K_{3,2} \} & =  \frac{1}{18} \bigl( -2 F_1^7 \K_{2,1} + 6 F_1^6 \K_{3,1} - 12 F_1^5 \K_{3,2} - 21 F_1^4 F_2 \K_{3,1} - 2 F_1^4 F_3 \K_{2,1} + 42 F_1^3 F_2^2 \K_{2,1} 
\\ &\quad  + 66 F_1^3 F_2 \K_{3,2} + 30 F_1^3 F_3 \K_{3,1} - 45 F_1^2 F_2^2 \K_{3,1} - 72 F_1^2 F_2 F_3 \K_{2,1} - 78 F_1^2 F_3 \K_{3,2} 
\\ &\quad + 48 F_1 F_2 F_3 \K_{3,1} + 40 F_1 F_3^2 \K_{2,1} + 18 F_2^3 \K_{3,1} - 6 F_2^2 F_3 \K_{2,1} + 24 F_2 F_3 \K_{3,2} - 36 F_3^2 \K_{3,1} \bigr)
\\ & \quad  +  \frac{1}{18} \imath h \chi^{-1} \biggl( \bigl( -F_1^8 + 20 F_1^6 F_2 - 16 F_1^5 F_3 - 78 F_1^4 F_2^2 + 128 F_1^3 F_2 F_3 + 36 F_1^2 F_2^3 
\\ & \quad  - 64 F_1^2 F_3^2 - 48 F_1 F_2^2 F_3 - 9 F_2^4 + 32 F_2 F_3^2 \bigr) \K_{2,1} - \frac{3}{2} \bigl( F_1^7 + 2 F_1^5 F_2 + 2 F_1^4 F_3 
\\ & \quad  - 33 F_1^3 F_2^2 + 46 F_1^2 F_2 F_3 + 6 F_1 F_2^3 - 24 F_1 F_3^2 \bigr) \K_{3,1} + \frac{9}{2} \biggl( F_1^6 - \frac{14}{3} F_1^4 F_2 
\\ & \quad  + \frac{16}{3} F_1^3 F_3 - F_1^2 F_2^2 - \frac{4}{3} F_1 F_2 F_3 - 2 F_2^3 + \frac{8}{3} F_3^2 \biggr) \K_{3,2} \biggr). 
\end{split}
\end{align}	
\end{subequations}  

Note that, as expected, $ F_1 $ is no longer a central element of the algebra, see~\cref{eq:ex:N=3-discr-allrels-F1-K21,eq:ex:N=3-discr-allrels-F1-K31,eq:ex:N=3-discr-allrels-F1-K32}. Moreover, by setting $h=0 $ in~eqs.~\eqref{eq:ex:N=3-discr-allrels}, we recover~\eqref{eq:ex:N=3-cont:all-rels} from~\cref{ex:N=3-cont}; hence the symmetry algebra $\widetilde{\Alg}^{(3)} $ can be indeed  understood as a deformation of $\Alg^{(3)} $ with respect to the time step $h$. Furthermore, the degree of $\widetilde{\Alg}^{(3)} $  is $10 $, provided by, for example,~\cref{eq:ex:N=3-discr-allrels-F3-K32}. This is in line with~\cref{prop:Symalg-N-discr,cor:deformation}. 

Finally, following~\cref{prop:discrete-additional-constants,prop:discr-closure-rels}, one can write a discrete-time functional relation among the elements of $\widetilde{\genset}^{(N)} $, analogous to~\eqref{eq:ex:N=3-cont:funrel}, as follows:    
\begin{equation}
	F_{3} \K_{3,1} - F_{2} \K_{3,2} - \left( \frac{1}{6} F_1^4 - F_1^2 F_2 + \frac{4}{3} F_1 F_3 + \frac{1}{2} F_2^2  \right) \K_{2,1} =0.   
\end{equation}
\end{example}

In summary, the deformation phenomenon constitutes a nontrivial feature of the model and its discretization.   In the concluding section, we will discuss the results and their implications, and compare them with the nonrelativistic case.

\section{Discussion and outlook}
\label{sec:concl}
In this work, we investigated the algebraic structures of the superintegrable rational Ruijsenaars--Schneider model~\cite{RS1986,AyadiFeher2010} and its integrable discretization, built in~\cite{NijhoffRagnisco1996}. Through the construction of additional integrals of motion, we showed that this discretization is superintegrable, similarly to its continuous counterpart. We determined the complete structure of the symmetry algebras of the continuous- and discrete-time $N$-body RS models, and demonstrated that the symmetry algebra in the discrete case is a deformation of the continuous one with respect to the time step $h$. 
Let us now compare our findings with the analogous results for the  (nonrelativistic) Calogero--Moser case; see~\cite{DGL2026Calogero}.

First, in the~\textbf{continuous-time case}, the structures of the symmetry algebras of the RS model and the Calogero--Moser model are similar to each other, up to a choice of the central element ($F_2 $ in the Calogero--Moser case, corresponding to the quadratic (in momenta) Hamiltonian function in Newtonian mechanics, and $F_1 $ in the relativistic case). Moreover, the same as in~\cite{DGL2026Calogero}, the algebra is determined by the structure relations and the closure relations involving Bell polynomials. This is in line with the note in~\cite[section IV]{DGL2026Calogero}, suggesting that similar formulas should hold true for other models that admit matrix traces as integrals of motion. However, while for $ N$-body $(N >2) $ Calogero--Moser model the degree of a polynomial algebra was $2N -1$, for the rational RS model the degree becomes $3N $, due to the relations~\eqref{eq:JFandJJ} and~\eqref{eq:prop:cont-formal-rels}. This fact suggests that there might exist a class of generalized superintegrable systems with $N $ degrees of freedom with symmetry algebras of higher degree, and both Calogero--Moser and rational RS systems are members of this class. 

Furthermore, although the RS model is considered as ``relativistic'' in a special sense (see~\cite{Braden1997}  and~\cref{sec:bg-RS-review}), this algebraic perspective might be used in future works to clarify certain aspects of other relativistic generalizations, connecting them to the existing developments in superintegrability in the relativistic setting; see~\cite{HeinzlIlderton2017,Ansell2018}. Note that the fact that the algebraic structures of the Calogero--Moser and RS models share common features is known: see, \eg{},~\cite{SurisWhy1996}. Our construction confirms that this can also be said about the symmetry algebras of the models in the context of superintegrability.  

Second, the modified integrals of motion, ensuring the maximal superintegrability of the discrete-time RS model, are constructed in a similar way as for the discretization of the Calogero--Moser system~\cite{nijhoffTimediscretizedVersionCalogeroMoser1994}: the construction follows the steps of~\cite{ujinoAdditionalConstantsMotion2008}, an extension of \cite{wojciechowski1983} to the discrete-time case. This allowed us to generalize the further procedure from~\cite{DGL2026Calogero} to the symmetry algebras of the RS model and its discretization. 

Third, using the modified integrals of motion, we obtained that the polynomial symmetry algebra~$\widetilde{\Alg}^{(N)} $ of the \textbf{discrete-time model} is a deformation of the symmetry algebra $\Alg^{(N)} $ of the continuous-time model with respect to the discretization parameter $h$, and the contraction $ \widetilde{\Alg}^{(N)} 
 \to  \Alg^{(N)}  
$ is given simply by taking the continuous limit $ h \to 0  $. Moreover, the degree of the polynomial algebra is also not preserved: it is increased from $3N $ to $3N +1$. Since the discrete Calogero--Moser model exhibits the same phenomenon, the results presented in this work can be considered as a relativistic generalization of~\cite{DGL2026Calogero}.  In both cases, the (formal) deformed structure relations are of the form:
\begin{subequations}
\begin{alignat}{3}
 \{ F_k, \K_{m,n} \} & = & \ \{ F_k, K_{m,n} \} - \varkappa (h) \,   
	& \Psi^{(1)}_{k,m,n \phantom{,i} } ( F_1, \ldots, F_N ),
 \\ 
 \{ \K_{i,j} , \K_{m,n} \} & = & \  \{ K_{i,j} , K_{m,n} \}  - \varkappa (h) \, 
	& \Psi^{(2)}_{i,j,m,n} ( F_1, \ldots, F_N, \K_{2,1}, \ldots, \K_{N, N-1} ),
\end{alignat} 
\end{subequations}
where $\Psi^{(1)}_{-,-,-} (\hyphen) $ and $\Psi^{(2)}_{-,-,-,-} (\hyphen) $ are scalar polynomial functions of their arguments, and:
\begin{itemize}
	\item $\varkappa (h)  \sim \sqrt{h} $ for the discrete-time Calogero--Moser model, 
	\item $ \varkappa (h) \sim h  $ for the discrete-time RS model.  
\end{itemize}
Therefore, we can observe that the relativistic generalization squares the deformation parameter, which follows from the direct construction of the integrals $\K_{m,n} $ in both cases. 

Finally, let us speculate on potential applications of the obtained results. As was mentioned in~\cref{sec:intro}, the construction of new discretizations retaining superintegrability is, in general, a challenging and delicate task. Here, in principle, the symmetry algebra formalism might provide a unifying framework for both continuous- and discrete-time systems. In particular, one may think about the inverse problem, \ie{}, can one build a new superintegrable discretization by deforming a polynomial symmetry algebra of a given continuous-time system in a systematic manner? However,  any attempt to develop such a theory requires many more examples than~\cite{DGL2026Calogero} and the present work. In this regard, further directions include studying other superintegrable systems and their symmetry algebras in the context of time-discretization, since this is expected to provide new insights into the algebraic approach to the discretization problem itself.  

\newpage

\section*{Acknowledgements}

The author is grateful to \mbox{Prof. Giorgio Gubbiotti, Dr. Danilo Latini,  and Prof. Stefano {Ansoldi}} for reading the final draft of the manuscript and providing valuable feedback, and Prof.~L\'aszl\'o Feh\'er for his insightful comments on the first version of the preprint. Moreover, I~thank Prof.~Naruhiko Aizawa for his hospitality during my research visit at Osaka Metropolitan University and for the opportunity to present these results at the workshop at~OMU.

This work was supported by  
the research project Mathematical
Methods in NonLinear Physics (MMNLP), 
Gruppo-4 Fisica Teorica of INFN, and partially by the National Group of Mathematical Physics (GNFM) 
of the Italian 
Institute for High Mathematics (INdAM). 
The author also acknowledges 
the support of the Ph.D.~program of the
Universit\`a degli Studi di~Udine.

\footnotesize 
\bibliographystyle{alphaurl} 

\bibliography{bib_dRS.bib}

\end{document}